\pdfoutput=1

\documentclass{eptcs}
\usepackage{hhline}
\usepackage{bookmark}
\usepackage[table]{xcolor}
\usepackage[all,cmtip]{xy}

\usepackage{tikzit}

\tikzstyle{H}=[draw, color=black, fill={rgb:black,1;white,3}, shape=rectangle, tikzit fill=yellow]
\tikzstyle{thick}=[draw, line width=0.5mm]
\tikzstyle{Z}=[draw, fill=white, circle, scale=1, inner sep=0pt, minimum size=10pt]
\tikzstyle{X}=[draw, fill={rgb:black,1;white,3}, text=black, circle, scale=1, inner sep=0pt, minimum size=10pt, tikzit fill={rgb,255: red,191; green,191; blue,191}]
\tikzstyle{Xthick}=[draw, fill=white, circle, scale=1, inner sep=0pt, minimum size=15pt, line width=0.5mm, tikzit fill={rgb,255: red,191; green,191; blue,191}]
\tikzstyle{Zthick}=[draw, fill={rgb:black,1;white,3}, text=black, circle, scale=1, inner sep=0pt, minimum size=15pt, line width=0.5mm]
\tikzstyle{phase}=[draw, fill=white, diamond, scale=1, inner sep=0pt, minimum size=10pt]
\tikzstyle{discard}=[draw, xscale=2.2, ground, rotate=90]
\tikzstyle{mmixed}=[draw, quantum, yscale=-2.2, ground, rotate=180]
\tikzstyle{quantum}=[line width=.6mm]
\tikzstyle{map}=[draw, color=black, fill=white, rectangle]
\tikzstyle{mapperp}=[draw, color=white, fill=black, text=white, rectangle]
\tikzstyle{s}=[draw, color=black, fill=gray, rectangle]
\tikzstyle{mapthick}=[draw, color=black, fill=white, rectangle, inner sep=0pt, minimum size=15pt, line width=0.5mm]
\tikzstyle{otimes}=[draw, fill=white, rotate=45, scale=0.9, minimum height=.1cm, circle, append after command={[shorten >=\pgflinewidth, shorten <=\pgflinewidth]
(\tikzlastnode.north) edge (\tikzlastnode.south)
(\tikzlastnode.east) edge (\tikzlastnode.west)}]
\tikzstyle{dot}=[thick, fill=black, circle, scale=1, inner sep=.05cm]
\tikzstyle{oplus}=[draw, scale=0.9, minimum height=.1cm, circle, append after command={[shorten >=\pgflinewidth, shorten <=\pgflinewidth]
(\tikzlastnode.north) edge (\tikzlastnode.south)
(\tikzlastnode.east) edge (\tikzlastnode.west)
}]
\tikzstyle{andin}=[draw, and gate US, rotate=90, scale=1, fill=white, label={center:{\it \&}}]
\tikzstyle{mulin}=[draw, and gate US, rotate=90, scale=1, fill=white, label={center:{}}]
\tikzstyle{andout}=[draw, and gate US, rotate=-90, scale=1, fill=white, label={center:{\it \&}}]
\tikzstyle{scalar}=[draw, rounded corners=1ex, rectangle round south west=false, rectangle round south east=false, tikzit fill={rgb,255: red,129; green,253; blue,255}, tikzit shape=rectangle]
\tikzstyle{scalarop}=[draw, rounded corners=1ex, rectangle round north west=false, rectangle round north east=false, tikzit fill={rgb,255: red,116; green,172; blue,255}, tikzit shape=rectangle]
\tikzstyle{fanin}=[draw, shape border rotate=30, regular polygon, regular polygon sides=3, fill=white, inner sep=.1cm]
\tikzstyle{fanout}=[draw, shape border rotate=-30, regular polygon, regular polygon sides=3, fill=white, inner sep=.1cm]
\tikzstyle{onein}=[draw, shape border rotate=30, regular polygon, regular polygon sides=3, fill=black, inner sep=.04cm, scale=1.2]
\tikzstyle{oneout}=[draw, shape border rotate=-30, regular polygon, regular polygon sides=3, fill=black, inner sep=.04cm, scale=1.2]
\tikzstyle{zeroin}=[draw, shape border rotate=30, regular polygon, regular polygon sides=3, fill=white, inner sep=.04cm, scale=1.2]
\tikzstyle{zeroout}=[draw, shape border rotate=-30, regular polygon, regular polygon sides=3, fill=white, inner sep=.04cm, scale=1.2]

\input{lagrel.tikzdefs}

\usepackage{mdframed}
\usepackage{arydshln}
\usepackage{multicol}
\renewcommand{\tilde}{\widetilde}
\usepackage{everypage}
\usepackage{lipsum}
\usepackage{amsthm}
\usepackage[inline]{enumitem}   
\usepackage{scalerel,stackengine}
\stackMath
\renewcommand\hat[1]{%
\savestack{\tmpbox}{\stretchto{%
  \scaleto{%
    \scalerel*[\widthof{\ensuremath{#1}}]{\kern-.6pt\bigwedge\kern-.6pt}%
    {\rule[-\textheight/2]{1ex}{\textheight}}
  }{\textheight}%
}{0.5ex}}%
\stackon[1pt]{#1}{\tmpbox}%
}
\newcommand{\bcell}{\cellcolor{black!10}}

\makeatletter
\newcommand{\inlineitem}[1][]{%
\ifnum\enit@type=\tw@
    {\descriptionlabel{#1}}
  \hspace{\labelsep}%
\else
  \ifnum\enit@type=\z@
       \refstepcounter{\@listctr}\fi
    \quad\@itemlabel\hspace{\labelsep}%
\fi}
\makeatother
\newcommand{\alr}{{\sf alr}}
\newcommand{\lr}{{\sf lr}}
\newcommand{\rel}{{\sf r}}
\newcommand{\aih}{{\sf aih}}
\newcommand{\ih}{{\sf ih}}

\newcommand{\xrightarrowtail}[1]{\!\!{\xymatrix@C=1em{\ar@{>->}[r]^{#1}&}}\!\!\!}
\newcommand{\xleftarrowtail}[1]{\!\!\!{\xymatrix@C=1em{&\ar@{>->}[l]_{#1}}}\!\!}

\newcommand{\xrightarrowiso}[1]{\!\!{\xymatrix@C=1em{\ar@{->}[r]^{#1}_\cong&}}\!\!\!}
\newcommand{\xleftarrowiso}[1]{\!\!\!{\xymatrix@C=1em{&\ar@{->}[l]_{#1}^\cong}}\!\!}

  \newtheorem{theorem}{Theorem}[section]
  \newtheorem{corollary}[theorem]{Corollary}
  \newtheorem{lemma}[theorem]{Lemma}
  \newtheorem{proposition}[theorem]{Proposition}

  \newtheorem{definition}[theorem]{Definition}
  \newtheorem{example}[theorem]{Example}
  
  \newtheorem{remark}[theorem]{Remark}

\newcommand{\Mat}{\mathsf{Mat}}

\newcommand{\cnot}{\mathsf{cnot}}
\newcommand{\tof}{\mathsf{tof}}
\newcommand{\Not}{\mathsf{not}}
\newcommand{\zeroin}{|0\rangle}
\newcommand{\zeroout}{\langle 0|}
\newcommand{\CNOT}{\mathsf{CNOT}}
\newcommand{\Sets}{\mathsf{Set}}
\newcommand{\FSets}{\mathsf{FinOrd}}
\newcommand{\FinOrd}{\mathsf{FinOrd}}
\newcommand{\TOF}{\mathsf{TOF}}
\newcommand{\Span}{\mathsf{Span}}
\newcommand{\dec}{\mathsf{dec}}
\newcommand{\Rel}{\mathsf{Rel}}
\newcommand{\op}{\mathsf{op}}
\newcommand{\co}{\mathsf{co}}
\newcommand{\Hilb}{\mathsf{Hilb}}
\newcommand{\FdHilb}{\mathsf{FHilb}}
\newcommand{\FHilb}{\mathsf{FHilb}}
\newcommand{\CPM}{\mathsf{CPM}}
\newcommand{\CP}{\mathsf{CP}}
\newcommand{\FPinj}{\mathsf{FPinj}}
\newcommand{\FPar}{\mathsf{FPar}}
\newcommand{\FSpan}{\mathsf{FSpan}}
\newcommand{\Pinj}{\mathsf{Pinj}}
\newcommand{\Par}{\mathsf{Par}}
\newcommand{\Aff}{\mathsf{Aff}}
\newcommand{\ParIso}{\mathsf{ParIso}}

\newcommand{\Total}{\mathsf{Total}}
\newcommand{\CFrob}{\mathsf{CFrob}}
\newcommand{\tr}{\mathsf{Tr}}
\newcommand{\ox}{\otimes}
\newcommand{\Csp}{{\sf Cospan}}
\newcommand{\Corel}{{\sf Corel}}
\newcommand{\Bool}{\mathbb{B}}
\newcommand{\Iso}{{\sf Iso}}
\renewcommand{\P}{{\sf p}}
\newcommand{\pmul}{{\sf pmul}}

\newcommand{\Prof}{{\sf Prof}}
\newcommand{\Mod}{{\sf Mod}}

\newcommand{\unit}{{\sf unit}}
\newcommand{\comm}{{\sf comm}}
\newcommand{\assoc}{{\sf assoc}}
\newcommand{\inj}{{\sf Inj}}
\newcommand{\surj}{{\sf Surj}}
\newcommand{\PSurj}{{\sf PSurj}}

\newcommand{\pre}{{\sf pre}}
\newcommand{\poly}{{\sf poly}}
\newcommand{\sub}{{\sf sub}}

\newcommand{\C}{\mathbb{C}}

\newcommand{\m}{{\sf m}}
\newcommand{\cm}{{\sf cm}}
\newcommand{\cb}{{\sf cb}}
\newcommand{\pcm}{{\sf pcm}}
\renewcommand{\r}{{\sf r}}
\newcommand{\scfrob}{{\sf scfrob}}

\newcommand{\bi}{{\sf b1}}
\newcommand{\bii}{{\sf b2}}
\newcommand{\biii}{{\sf b3}}
\newcommand{\biv}{{\sf b4}}

\newcommand{\Kl}{{\sf Kl}}
\newcommand{\Mon}{{\sf Mon}}

\newcommand{\F}{\mathbb{F}}
\newcommand{\X}{\mathbb{X}}
\newcommand{\Y}{\mathbb{Y}}
\newcommand{\Z}{\mathbb{Z}}
\newcommand{\N}{\mathbb{N}}
\newcommand{\T}{\mathbb{T}}
\newcommand{\s}{\mathbb{S}}
\newcommand{\U}{\mathbb{U}}

\newcommand{\IH}{\mathbb{IH}}

\newcommand{\M}{\mathcal{M}}
\newcommand{\E}{\mathcal{E}}

\newcommand{\eq}[1]{\stackrel{\scalebox{.6}{#1}}{=}}

\newcommand{\defeq}[1]{\stackrel{\scalebox{.6}{#1}}{:=}}

\newcommand{\ZXA}{\mathsf{ZX}\textit{\&}}

\newcommand{\Vect}{\mathsf{Vect}}
\newcommand{\Lag}{\mathsf{Lag}}
\newcommand{\im}{\mathsf{im}}
\newcommand{\ZX}{\mathsf{ZX}}
\newcommand{\ZH}{\mathsf{ZH}}
\DeclareMathSymbol{\bot}{\mathord}{symbols}{"3F}

\newcommand{\pullbackcorner}[1][dl]{\save*!/#1-1pc/#1:(-1,1)@^{|-}\restore}

\renewcommand{\epsilon}{\varepsilon}
\renewcommand{\bar}[1]{\overline{#1}\hspace*{.01cm}}

\newcommand{\Stab}{{\sf Stab}}
\newcommand{\LinRel}{\sf LinRel}

\usepackage{amsmath}
\usepackage{pict2e}

\newcommand{\lbparen}{\{
}

\newcommand{\rbparen}{ \}
}

\newdir{|>}{-<5pt,0pt>{
\begin{tikzpicture}[scale=.7]
	\begin{pgfonlayer}{nodelayer}
		\node [style=none] (0) at (0, 0) {};
		\node [style=none] (1) at (1, 0) {};
		\node [style=none] (2) at (-1, -0.25) {};
	\end{pgfonlayer}
	\begin{pgfonlayer}{edgelayer}
		\draw (2.center) to (0.center);
		\draw (0.center) to (1.center);
	\end{pgfonlayer}
\end{tikzpicture}
}}
\newdir{|<}{-<5pt,0pt>{
\begin{tikzpicture}[scale=.9]
	\begin{pgfonlayer}{nodelayer}
		\node [style=none] (0) at (0, -0.25) {};
		\node [style=none] (1) at (-1, -0.25) {};
		\node [style=none] (2) at (1, 0) {};
	\end{pgfonlayer}
	\begin{pgfonlayer}{edgelayer}
		\draw (2.center) to (0.center);
		\draw (0.center) to (1.center);
	\end{pgfonlayer}
\end{tikzpicture}
}}

\newcommand{\skewpullbackcorner}[1][dl]{\save*!/#1-1.1pc/#1:(-.5,1)@^{|>}\restore}
\newcommand{\skewpushoutcorner}[1][dl]{\save*!/#1-1pc/#1:(-1,1)@^{|<}\restore}

\DeclareFontFamily{U}{mathx}{\hyphenchar\font45}
\DeclareFontShape{U}{mathx}{m}{n}{
      <5> <6> <7> <8> <9> <10>
      <10.95> <12> <14.4> <17.28> <20.74> <24.88>
      mathx10
      }{}
\DeclareSymbolFont{mathx}{U}{mathx}{m}{n}
\DeclareFontSubstitution{U}{mathx}{m}{n}
\DeclareMathAccent{\widecheck}{0}{mathx}{"71}
\DeclareMathAccent{\wideparen}{0}{mathx}{"75}

\def\cs#1{\texttt{\char`\\#1}}

\usepackage{amsmath}

\usepackage{hyperref}

\newcommand\numeq[2]%
  {\label{#2}\stackrel{\scriptscriptstyle(\mkern-1.5mu#1\mkern-1.5mu)}{=}}

\xymatrixrowsep{.5cm}
\xymatrixcolsep{.65cm}

\title{A Graphical Calculus for Lagrangian Relations}
\def\titlerunning{A Graphical Calculus for Lagrangian Relations}
\date{\today}
\author{Cole Comfort, Aleks Kissinger\\ Department of Computer Science, University of Oxford}
\def\authorrunning{Cole Comfort, Aleks Kissinger}

\begin{document}

\maketitle

\begin{abstract}
Symplectic vector spaces are the phase spaces of linear mechanical systems.  The symplectic form describes, for example, the relation between position and momentum as well as current and voltage. The category of linear Lagrangian relations between symplectic vector spaces is a symmetric monoidal subcategory of relations which gives a semantics for the evolution --- and more generally linear constraints on the evolution --- of various physical systems.

We give a new presentation of the category of Lagrangian relations over an arbitrary field as a `doubled' category of linear relations. More precisely, we show that it arises as a variation of Selinger's CPM construction applied to linear relations, where the covariant orthogonal complement functor plays the role of conjugation. Furthermore, for linear relations over prime fields, this corresponds exactly to the CPM construction for a suitable choice of dagger. We can furthermore extend this construction by a single affine shift operator to obtain a category of affine Lagrangian relations. Using this new presentation, we prove the equivalence of the prop of affine Lagrangian relations with the prop of qudit stabilizer theory in odd prime dimensions. We hence obtain a unified graphical language for several disparate process theories, including electrical circuits, Spekkens' toy theory, and odd-prime-dimensional stabilizer quantum circuits.
\end{abstract}

Linear Lagrangian relations, or more generally, affine Lagrangian relations provide a rich, compositional setting for modelling the evolutions of various physical systems. For example, certain classes of electrical circuits can be interpreted in terms of Lagrangian relations over the field of real rational functions~\cite{passive,affine}. On a quite different note, the state preparation and quantum evolution of $p$-dimensional generalizations of Spekkens' toy theory~\cite{spekkens2016quasi} and (consequently) odd-prime-dimensional stabilizer quantum theory~\cite{gross} have semantics in terms of affine Lagrangian relations over $\F_p$.  Specifically, the state preparation corresponds to the affine Lagrangian relations from the tensor unit, and the evolution corresponds to affine symplectomorphisms.  In this paper we extend this correspondance to the full category of Lagrangian relations, giving these circuits a proper categorical treatment.

Formally, the category of Lagrangian relations is the symmetric monoidal subcategory of linear relations where the objects are symplectic vector spaces and the morphisms are linear relations satisfying an extra condition which can be captured graphically as the following, where $V^\perp$ denotes the orthogonal complement and the grey box denotes the \textit{antipode} from the graphical theory of linear relations:
$$
\begin{tikzpicture}
	\begin{pgfonlayer}{nodelayer}
		\node [style=map] (0) at (2, -2) {$V$};
		\node [style=none] (1) at (1.75, -1.25) {};
		\node [style=none] (2) at (2.25, -1.25) {};
		\node [style=none] (3) at (1.75, -2.75) {};
		\node [style=none] (4) at (2.25, -2.75) {};
	\end{pgfonlayer}
	\begin{pgfonlayer}{edgelayer}
		\draw [in=120, out=-90] (1.center) to (0);
		\draw [in=-90, out=60] (0) to (2.center);
		\draw [in=-60, out=90] (4.center) to (0);
		\draw [in=90, out=-120] (0) to (3.center);
	\end{pgfonlayer}
\end{tikzpicture}
=
\begin{tikzpicture}
	\begin{pgfonlayer}{nodelayer}
		\node [style=map] (0) at (2, -2) {$V^\perp$};
		\node [style=none] (1) at (1.75, -1.25) {};
		\node [style=none] (2) at (2.25, -1.25) {};
		\node [style=none] (3) at (1.75, -2.75) {};
		\node [style=none] (4) at (2.25, -2.75) {};
		\node [style=none] (5) at (2.25, -0.5) {};
		\node [style=none] (6) at (1.75, -0.5) {};
		\node [style=none] (7) at (2.25, -3.5) {};
		\node [style=none] (8) at (1.75, -3.5) {};
		\node [style=s] (9) at (2.25, -1.25) {};
		\node [style=s] (10) at (2.25, -2.75) {};
	\end{pgfonlayer}
	\begin{pgfonlayer}{edgelayer}
		\draw [in=120, out=-90] (1.center) to (0);
		\draw [in=-90, out=60] (0) to (2.center);
		\draw [in=-60, out=90] (4.center) to (0);
		\draw [in=90, out=-120] (0) to (3.center);
		\draw [in=90, out=-90] (6.center) to (2.center);
		\draw [in=270, out=90] (1.center) to (5.center);
		\draw [in=270, out=90] (7.center) to (3.center);
		\draw [in=270, out=90] (8.center) to (4.center);
	\end{pgfonlayer}
\end{tikzpicture}
$$
We show that any linear relation $V$ determines a Lagrangian relation in terms of `doubling', i.e.\ taking the tensor product of a linear relation with its complement:
$$
\begin{tikzpicture}
	\begin{pgfonlayer}{nodelayer}
		\node [style=map] (24) at (3.25, -1) {$V$};
		\node [style=none] (25) at (3.25, -0.25) {};
		\node [style=none] (27) at (3.25, -1.75) {};
	\end{pgfonlayer}
	\begin{pgfonlayer}{edgelayer}
		\draw (25.center) to (24);
		\draw (27.center) to (24);
	\end{pgfonlayer}
\end{tikzpicture}
\mapsto
\begin{tikzpicture}
	\begin{pgfonlayer}{nodelayer}
		\node [style=map] (24) at (3.25, -1) {$V^\perp$};
		\node [style=none] (25) at (3.25, -0.25) {};
		\node [style=none] (27) at (3.25, -1.75) {};
		\node [style=map] (28) at (4, -1) {$V$};
		\node [style=none] (29) at (4, -0.25) {};
		\node [style=none] (30) at (4, -1.75) {};
	\end{pgfonlayer}
	\begin{pgfonlayer}{edgelayer}
		\draw (25.center) to (24);
		\draw (27.center) to (24);
		\draw (29.center) to (28);
		\draw (30.center) to (28);
	\end{pgfonlayer}
\end{tikzpicture}
$$
By analogy to the CPM construction for the category of completely positive maps, we call these \textit{pure} Lagrangian relations. In Theorem \ref{theorem:unbiased} we show that only one more class of `discard' generators $d_a$ for each $a$ in the underlying field $k$ is required to generate all Lagrangian relations.
$$
d_a := 
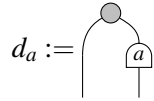
\begin{tikzpicture}
	\begin{pgfonlayer}{nodelayer}
		\node [style=X] (0) at (0, 0.75) {};
		\node [style=scalar] (1) at (0.5, 0) {$a$};
		\node [style=none] (2) at (0.5, -0.75) {};
		\node [style=none] (3) at (-0.5, 0) {};
		\node [style=none] (4) at (-0.5, -0.75) {};
	\end{pgfonlayer}
	\begin{pgfonlayer}{edgelayer}
		\draw [in=-30, out=90] (1) to (0);
		\draw [in=90, out=-150] (0) to (3.center);
		\draw (4.center) to (3.center);
		\draw (2.center) to (1);
	\end{pgfonlayer}
\end{tikzpicture}

$$

From this, we immediately obtain a complete graphical calculus for Lagrangian relations over any field $k$, namely we can apply the complete calculus $\ih_k$ for linear relations~\cite{ihpub} to diagrams built from pure morphisms and discard maps. 
This extends the doubled presentation of bond graphs, given in \cite[5.3]{coya}, which are not universal for Lagrangian relations, and is instead only universal for a fragment of the pure morphisms.
In Corollary \ref{cor:pure}, we also immediately get a \textit{purification theorem} for Lagrangian relations, much like the purification (a.k.a. Stinespring dilation) of quantum channels which can be proven straightforwardly in the CPM construction over Hilbert spaces.

Furthermore, in the case of prime fields, i.e. finite fields $\mathbb F_p$ for $p$, we show in Corollary \ref{cor} that this is actually an instance of the original CPM construction, for a suitably defined dagger on the category of linear relations.

In Section \ref{sec:aff} we show that only one more generator is needed to obtain {\em affine} Lagrangian relations.  In the case of odd prime fields, we show in Theorem \ref{theorem:spekkens} that affine Lagrangian relations are prime-dimensional qudit stabilizer circuits, modulo invertible scalars.  This give a graphical calculus extends previous work on the qubit \cite{backensspek}, and qutrit \cite{qutrit} cases.  We also discuss the relation to electrical circuits.

\paragraph{Related work.} It was previously shown that certain classes of electrical circuits have a semantics in terms of affine Lagrangian relations over the field of the real numbers and the real rational functions ${\mathbb{R}[x,y]/\langle xy-1\rangle}$ \cite{network,passive}. Similarly in \cite[\S VI]{affine}, the authors give an interpretation of non-passive electrical circuits in terms of these `doubled' string diagrams for affine relations over the real rational functions. However the authors did not give a full characterisation for the category of Lagrangian relations in terms of diagrammatic generators.
We restate the interpretations of the electrical components given in  \cite[\S VI]{affine} in terms of the graphical calculus for affine Lagrangian relations in Example \ref{ex:circuits}.

A presentation of odd-prime stabilizer theory in terms of affine symplectomorphisms applied to Lagrangian subspaces appears in~\cite{gross} and several follow-on works relating stabilizer theory to classical phase spaces via the discrete Wigner function. Our Theorem \ref{theorem:spekkens} is a categorical reformulation of the result of Spekkens' in which he shows that so called odd-prime-dimensional `quadrature epistricted theories' are operationally equivalent to prime-dimensional qudit stabilizer circuits \cite{spekkens2016quasi}, following earlier work in \cite{spekkens}.  This operational equivalence has also been further explored in the non-prime case \cite{catani}. Note that operational equivalence is not the same as categorical equivalence. The notion of operational equivalence used in \cite{spekkens2016quasi,catani}  refers to the equivalence of protocols in which circuits are prepared, evolve and then are measured, whereas ours is more `process-theoretic', i.e.\ we consider the category that contains states, effects, evolutions, and all possible compositions thereof. A complete presentation for Spekkens' qubit toy model in terms of a category of relations has also been given \cite{backensspek} following the categorical description by \cite{coecke2012spekkens}.  However, the authors do not explicitly establish that this is the category of affine Lagrangian relations over $\F_2$, but merely a subcategory of finite sets and relations. There is also a complete presentation for qutrit stabilizer theory \cite{qutrit} which, by Theorem \ref{theorem:spekkens}, is equivalent to Spekkens' qutrit toy model, up to scalars; the connection to relations, in this case, is unexplored.

\section{Linear relations}

\label{sec:linear}

In order to decribe Lagrangian relations diagramatically, we must first recall the symmetric monoidal theory of linear relations.  To do so, we first recall the symmetric monoidal theory of matrices:

\begin{definition}[{\cite[Defn.\ 3.4]{ih}}]
Given a ring $k$, let $\cb_k$ denote the prop given by the generators\footnote{We use the ZX-style colouring which is dual to that used in \cite{ih}.}:
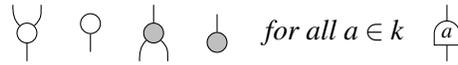
$$
\begin{tikzpicture}
	\begin{pgfonlayer}{nodelayer}
		\node [style=Z] (0) at (4, 2) {};
		\node [style=none] (1) at (3.75, 2.5) {};
		\node [style=none] (2) at (4.25, 2.5) {};
		\node [style=none] (3) at (4, 1.5) {};
	\end{pgfonlayer}
	\begin{pgfonlayer}{edgelayer}
		\draw [in=135, out=-90] (1.center) to (0);
		\draw (0) to (3.center);
		\draw [in=-90, out=45, looseness=1.25] (0) to (2.center);
	\end{pgfonlayer}
\end{tikzpicture}\hspace*{,5cm}
\begin{tikzpicture}
	\begin{pgfonlayer}{nodelayer}
		\node [style=Z] (0) at (4, 2) {};
		\node [style=none] (3) at (4, 1.5) {};
	\end{pgfonlayer}
	\begin{pgfonlayer}{edgelayer}
		\draw (0) to (3.center);
	\end{pgfonlayer}
\end{tikzpicture}\hspace*{,5cm}
\begin{tikzpicture}
	\begin{pgfonlayer}{nodelayer}
		\node [style=X] (0) at (4, 2) {};
		\node [style=none] (1) at (3.75, 1.5) {};
		\node [style=none] (2) at (4.25, 1.5) {};
		\node [style=none] (3) at (4, 2.5) {};
	\end{pgfonlayer}
	\begin{pgfonlayer}{edgelayer}
		\draw [in=-135, out=90] (1.center) to (0);
		\draw (0) to (3.center);
		\draw [in=90, out=-45, looseness=1.25] (0) to (2.center);
	\end{pgfonlayer}
\end{tikzpicture}\hspace*{,5cm}
\begin{tikzpicture}
	\begin{pgfonlayer}{nodelayer}
		\node [style=X] (0) at (4, 2) {};
		\node [style=none] (3) at (4, 2.5) {};
	\end{pgfonlayer}
	\begin{pgfonlayer}{edgelayer}
		\draw (0) to (3.center);
	\end{pgfonlayer}
\end{tikzpicture}\hspace*{,5cm}
\text{for all $a \in k$ }
\hspace*{,3cm}
\begin{tikzpicture}
	\begin{pgfonlayer}{nodelayer}
		\node [style=scalar] (0) at (4, 2) {$a$};
		\node [style=none] (3) at (4, 2.5) {};
		\node [style=none] (4) at (4, 1.5) {};
	\end{pgfonlayer}
	\begin{pgfonlayer}{edgelayer}
		\draw (0) to (3.center);
		\draw (4.center) to (0);
	\end{pgfonlayer}
\end{tikzpicture}
$$
modulo the equations of a bicommutative bialgebra:
$$
\begin{tikzpicture}
	\begin{pgfonlayer}{nodelayer}
		\node [style=Z] (0) at (5, 2) {};
		\node [style=none] (1) at (4.75, 2.5) {};
		\node [style=none] (2) at (5.25, 2.5) {};
		\node [style=none] (3) at (5, 1.5) {};
	\end{pgfonlayer}
	\begin{pgfonlayer}{edgelayer}
		\draw (3.center) to (0);
		\draw [in=-90, out=135] (0) to (1.center);
		\draw [in=-90, out=45] (0) to (2.center);
	\end{pgfonlayer}
\end{tikzpicture}
=
\begin{tikzpicture}
	\begin{pgfonlayer}{nodelayer}
		\node [style=Z] (0) at (5, 2) {};
		\node [style=none] (1) at (4.75, 2.5) {};
		\node [style=none] (2) at (5.25, 2.5) {};
		\node [style=none] (3) at (5, 1.5) {};
		\node [style=none] (4) at (5.25, 3) {};
		\node [style=none] (5) at (4.75, 3) {};
	\end{pgfonlayer}
	\begin{pgfonlayer}{edgelayer}
		\draw (3.center) to (0);
		\draw [in=-90, out=135] (0) to (1.center);
		\draw [in=-90, out=45] (0) to (2.center);
		\draw [in=270, out=90] (1.center) to (4.center);
		\draw [in=270, out=90] (2.center) to (5.center);
	\end{pgfonlayer}
\end{tikzpicture}
\hspace*{.5cm}
\begin{tikzpicture}
	\begin{pgfonlayer}{nodelayer}
		\node [style=Z] (0) at (5, 2) {};
		\node [style=none] (1) at (4.75, 2.5) {};
		\node [style=none] (2) at (5.25, 2.5) {};
		\node [style=none] (3) at (5, 1.5) {};
		\node [style=none] (4) at (4.75, 3) {};
		\node [style=Z] (5) at (5.25, 2.5) {};
	\end{pgfonlayer}
	\begin{pgfonlayer}{edgelayer}
		\draw (3.center) to (0);
		\draw [in=-90, out=135] (0) to (1.center);
		\draw [in=-90, out=45] (0) to (2.center);
		\draw [in=270, out=90] (1.center) to (4.center);
	\end{pgfonlayer}
\end{tikzpicture}
=
\begin{tikzpicture}
	\begin{pgfonlayer}{nodelayer}
		\node [style=none] (1) at (4.75, 1.5) {};
		\node [style=none] (4) at (4.75, 3) {};
	\end{pgfonlayer}
	\begin{pgfonlayer}{edgelayer}
		\draw (1.center) to (4.center);
	\end{pgfonlayer}
\end{tikzpicture}
\hspace*{.5cm}
\begin{tikzpicture}
	\begin{pgfonlayer}{nodelayer}
		\node [style=Z] (0) at (5, 2) {};
		\node [style=none] (1) at (4.75, 2.5) {};
		\node [style=none] (2) at (5.25, 2.5) {};
		\node [style=none] (3) at (5, 1.5) {};
		\node [style=Z] (4) at (4.75, 2.5) {};
		\node [style=none] (5) at (4.5, 3) {};
		\node [style=none] (6) at (5, 3) {};
		\node [style=none] (7) at (5.25, 3) {};
	\end{pgfonlayer}
	\begin{pgfonlayer}{edgelayer}
		\draw (3.center) to (0);
		\draw [in=-90, out=135] (0) to (1.center);
		\draw [in=-90, out=45] (0) to (2.center);
		\draw [in=-90, out=135] (4) to (5.center);
		\draw [in=-90, out=45] (4) to (6.center);
		\draw (2.center) to (7.center);
	\end{pgfonlayer}
\end{tikzpicture}
=
\begin{tikzpicture}[xscale=-1]
	\begin{pgfonlayer}{nodelayer}
		\node [style=Z] (0) at (5, 2) {};
		\node [style=none] (1) at (4.75, 2.5) {};
		\node [style=none] (2) at (5.25, 2.5) {};
		\node [style=none] (3) at (5, 1.5) {};
		\node [style=Z] (4) at (4.75, 2.5) {};
		\node [style=none] (5) at (4.5, 3) {};
		\node [style=none] (6) at (5, 3) {};
		\node [style=none] (7) at (5.25, 3) {};
	\end{pgfonlayer}
	\begin{pgfonlayer}{edgelayer}
		\draw (3.center) to (0);
		\draw [in=-90, out=135] (0) to (1.center);
		\draw [in=-90, out=45] (0) to (2.center);
		\draw [in=-90, out=135] (4) to (5.center);
		\draw [in=-90, out=45] (4) to (6.center);
		\draw (2.center) to (7.center);
	\end{pgfonlayer}
\end{tikzpicture}
\hspace*{.5cm}
\begin{tikzpicture}[yscale=-1]
	\begin{pgfonlayer}{nodelayer}
		\node [style=X] (0) at (5, 2) {};
		\node [style=none] (1) at (4.75, 2.5) {};
		\node [style=none] (2) at (5.25, 2.5) {};
		\node [style=none] (3) at (5, 1.5) {};
	\end{pgfonlayer}
	\begin{pgfonlayer}{edgelayer}
		\draw (3.center) to (0);
		\draw [in=-90, out=135] (0) to (1.center);
		\draw [in=-90, out=45] (0) to (2.center);
	\end{pgfonlayer}
\end{tikzpicture}
=
\begin{tikzpicture}[yscale=-1]
	\begin{pgfonlayer}{nodelayer}
		\node [style=X] (0) at (5, 2) {};
		\node [style=none] (1) at (4.75, 2.5) {};
		\node [style=none] (2) at (5.25, 2.5) {};
		\node [style=none] (3) at (5, 1.5) {};
		\node [style=none] (4) at (5.25, 3) {};
		\node [style=none] (5) at (4.75, 3) {};
	\end{pgfonlayer}
	\begin{pgfonlayer}{edgelayer}
		\draw (3.center) to (0);
		\draw [in=-90, out=135] (0) to (1.center);
		\draw [in=-90, out=45] (0) to (2.center);
		\draw [in=270, out=90] (1.center) to (4.center);
		\draw [in=270, out=90] (2.center) to (5.center);
	\end{pgfonlayer}
\end{tikzpicture}
\hspace*{.5cm}
\begin{tikzpicture}[yscale=-1]
	\begin{pgfonlayer}{nodelayer}
		\node [style=X] (0) at (5, 2) {};
		\node [style=none] (1) at (4.75, 2.5) {};
		\node [style=none] (2) at (5.25, 2.5) {};
		\node [style=none] (3) at (5, 1.5) {};
		\node [style=none] (4) at (4.75, 3) {};
		\node [style=X] (5) at (5.25, 2.5) {};
	\end{pgfonlayer}
	\begin{pgfonlayer}{edgelayer}
		\draw (3.center) to (0);
		\draw [in=-90, out=135] (0) to (1.center);
		\draw [in=-90, out=45] (0) to (2.center);
		\draw [in=270, out=90] (1.center) to (4.center);
	\end{pgfonlayer}
\end{tikzpicture}
=
\begin{tikzpicture}
	\begin{pgfonlayer}{nodelayer}
		\node [style=none] (1) at (4.75, 1.5) {};
		\node [style=none] (4) at (4.75, 3) {};
	\end{pgfonlayer}
	\begin{pgfonlayer}{edgelayer}
		\draw (1.center) to (4.center);
	\end{pgfonlayer}
\end{tikzpicture}
\hspace*{.5cm}
\begin{tikzpicture}[yscale=-1]
	\begin{pgfonlayer}{nodelayer}
		\node [style=X] (0) at (5, 2) {};
		\node [style=none] (1) at (4.75, 2.5) {};
		\node [style=none] (2) at (5.25, 2.5) {};
		\node [style=none] (3) at (5, 1.5) {};
		\node [style=X] (4) at (4.75, 2.5) {};
		\node [style=none] (5) at (4.5, 3) {};
		\node [style=none] (6) at (5, 3) {};
		\node [style=none] (7) at (5.25, 3) {};
	\end{pgfonlayer}
	\begin{pgfonlayer}{edgelayer}
		\draw (3.center) to (0);
		\draw [in=-90, out=135] (0) to (1.center);
		\draw [in=-90, out=45] (0) to (2.center);
		\draw [in=-90, out=135] (4) to (5.center);
		\draw [in=-90, out=45] (4) to (6.center);
		\draw (2.center) to (7.center);
	\end{pgfonlayer}
\end{tikzpicture}
=
\begin{tikzpicture}[scale=-1]
	\begin{pgfonlayer}{nodelayer}
		\node [style=X] (0) at (5, 2) {};
		\node [style=none] (1) at (4.75, 2.5) {};
		\node [style=none] (2) at (5.25, 2.5) {};
		\node [style=none] (3) at (5, 1.5) {};
		\node [style=X] (4) at (4.75, 2.5) {};
		\node [style=none] (5) at (4.5, 3) {};
		\node [style=none] (6) at (5, 3) {};
		\node [style=none] (7) at (5.25, 3) {};
	\end{pgfonlayer}
	\begin{pgfonlayer}{edgelayer}
		\draw (3.center) to (0);
		\draw [in=-90, out=135] (0) to (1.center);
		\draw [in=-90, out=45] (0) to (2.center);
		\draw [in=-90, out=135] (4) to (5.center);
		\draw [in=-90, out=45] (4) to (6.center);
		\draw (2.center) to (7.center);
	\end{pgfonlayer}
\end{tikzpicture}
$$
$$
\begin{tikzpicture}
	\begin{pgfonlayer}{nodelayer}
		\node [style=Z] (0) at (5, 2) {};
		\node [style=X] (1) at (5, 1.5) {};
		\node [style=none] (2) at (4.75, 2.5) {};
		\node [style=none] (3) at (5.25, 2.5) {};
	\end{pgfonlayer}
	\begin{pgfonlayer}{edgelayer}
		\draw (1) to (0);
		\draw [in=-90, out=135] (0) to (2.center);
		\draw [in=-90, out=45] (0) to (3.center);
	\end{pgfonlayer}
\end{tikzpicture}
=
\begin{tikzpicture}
	\begin{pgfonlayer}{nodelayer}
		\node [style=none] (2) at (4.75, 2.5) {};
		\node [style=none] (3) at (5.25, 2.5) {};
		\node [style=X] (4) at (4.75, 2) {};
		\node [style=X] (5) at (5.25, 2) {};
	\end{pgfonlayer}
	\begin{pgfonlayer}{edgelayer}
		\draw (4) to (2.center);
		\draw (5) to (3.center);
	\end{pgfonlayer}
\end{tikzpicture}
\hspace*{.5cm}
\begin{tikzpicture}[yscale=-1]
	\begin{pgfonlayer}{nodelayer}
		\node [style=X] (0) at (5, 2) {};
		\node [style=Z] (1) at (5, 1.5) {};
		\node [style=none] (2) at (4.75, 2.5) {};
		\node [style=none] (3) at (5.25, 2.5) {};
	\end{pgfonlayer}
	\begin{pgfonlayer}{edgelayer}
		\draw (1) to (0);
		\draw [in=-90, out=135] (0) to (2.center);
		\draw [in=-90, out=45] (0) to (3.center);
	\end{pgfonlayer}
\end{tikzpicture}
=
\begin{tikzpicture}[yscale=-1]
	\begin{pgfonlayer}{nodelayer}
		\node [style=none] (2) at (4.75, 2.5) {};
		\node [style=none] (3) at (5.25, 2.5) {};
		\node [style=Z] (4) at (4.75, 2) {};
		\node [style=Z] (5) at (5.25, 2) {};
	\end{pgfonlayer}
	\begin{pgfonlayer}{edgelayer}
		\draw (4) to (2.center);
		\draw (5) to (3.center);
	\end{pgfonlayer}
\end{tikzpicture}
\hspace*{.5cm}
\begin{tikzpicture}
	\begin{pgfonlayer}{nodelayer}
		\node [style=Z] (0) at (5, 2) {};
		\node [style=X] (1) at (5, 1.5) {};
		\node [style=none] (2) at (4.75, 2.5) {};
		\node [style=none] (3) at (5.25, 2.5) {};
		\node [style=none] (4) at (4.75, 1) {};
		\node [style=none] (5) at (5.25, 1) {};
	\end{pgfonlayer}
	\begin{pgfonlayer}{edgelayer}
		\draw (1) to (0);
		\draw [in=-90, out=135] (0) to (2.center);
		\draw [in=-90, out=45] (0) to (3.center);
		\draw [in=-45, out=90] (5.center) to (1);
		\draw [in=90, out=-135] (1) to (4.center);
	\end{pgfonlayer}
\end{tikzpicture}
=
\begin{tikzpicture}
	\begin{pgfonlayer}{nodelayer}
		\node [style=Z] (0) at (5, 2) {};
		\node [style=Z] (1) at (5.5, 2) {};
		\node [style=X] (2) at (5, 2.5) {};
		\node [style=X] (3) at (5.5, 2.5) {};
		\node [style=none] (4) at (5.5, 3) {};
		\node [style=none] (5) at (5, 3) {};
		\node [style=none] (6) at (5, 1.5) {};
		\node [style=none] (7) at (5.5, 1.5) {};
	\end{pgfonlayer}
	\begin{pgfonlayer}{edgelayer}
		\draw (5.center) to (2);
		\draw [bend right=45, looseness=1.25] (2) to (0);
		\draw (0) to (3);
		\draw (1) to (2);
		\draw [bend left=45, looseness=1.25] (3) to (1);
		\draw (1) to (7.center);
		\draw (6.center) to (0);
		\draw (3) to (4.center);
	\end{pgfonlayer}
\end{tikzpicture}
\hspace*{.5cm}
\begin{tikzpicture}
	\begin{pgfonlayer}{nodelayer}
		\node [style=Z] (0) at (5, 2) {};
		\node [style=X] (1) at (5, 1.5) {};
	\end{pgfonlayer}
	\begin{pgfonlayer}{edgelayer}
		\draw (1) to (0);
	\end{pgfonlayer}
\end{tikzpicture}
=
$$
and the additional equations:
$$
\begin{tikzpicture}
	\begin{pgfonlayer}{nodelayer}
		\node [style=none] (1) at (5, 1.25) {};
		\node [style=scalar] (2) at (5, 1.75) {$a$};
		\node [style=Z] (3) at (5, 2.25) {};
		\node [style=none] (4) at (5.25, 2.75) {};
		\node [style=none] (5) at (4.75, 2.75) {};
	\end{pgfonlayer}
	\begin{pgfonlayer}{edgelayer}
		\draw (1.center) to (2);
		\draw (2) to (3);
		\draw [in=-90, out=135] (3) to (5.center);
		\draw [in=-90, out=45] (3) to (4.center);
	\end{pgfonlayer}
\end{tikzpicture}
=
\begin{tikzpicture}
	\begin{pgfonlayer}{nodelayer}
		\node [style=none] (0) at (6.5, 1.75) {};
		\node [style=Z] (1) at (6.5, 2.25) {};
		\node [style=scalar] (4) at (6.25, 2.75) {$a$};
		\node [style=scalar] (5) at (6.75, 2.75) {$a$};
		\node [style=none] (6) at (6.75, 3.25) {};
		\node [style=none] (7) at (6.25, 3.25) {};
	\end{pgfonlayer}
	\begin{pgfonlayer}{edgelayer}
		\draw (5) to (6.center);
		\draw (7.center) to (4);
		\draw (1) to (0.center);
		\draw [in=-90, out=45] (1) to (5);
		\draw [in=-90, out=135] (1) to (4);
	\end{pgfonlayer}
\end{tikzpicture}
\hspace*{.5cm}
\begin{tikzpicture}
	\begin{pgfonlayer}{nodelayer}
		\node [style=none] (0) at (5, 1.25) {};
		\node [style=scalar] (1) at (5, 1.75) {$a$};
		\node [style=Z] (2) at (5, 2.25) {};
	\end{pgfonlayer}
	\begin{pgfonlayer}{edgelayer}
		\draw (0.center) to (1);
		\draw (1) to (2);
	\end{pgfonlayer}
\end{tikzpicture}
=
\begin{tikzpicture}
	\begin{pgfonlayer}{nodelayer}
		\node [style=none] (0) at (5, 1.25) {};
		\node [style=Z] (2) at (5, 2.25) {};
	\end{pgfonlayer}
	\begin{pgfonlayer}{edgelayer}
		\draw (0.center) to (2);
	\end{pgfonlayer}
\end{tikzpicture}
\hspace*{.5cm}
\begin{tikzpicture}
	\begin{pgfonlayer}{nodelayer}
		\node [style=none] (0) at (5, 2.75) {};
		\node [style=scalar] (1) at (5, 2.25) {$a$};
		\node [style=X] (2) at (5, 1.75) {};
		\node [style=none] (3) at (5.25, 1.25) {};
		\node [style=none] (4) at (4.75, 1.25) {};
	\end{pgfonlayer}
	\begin{pgfonlayer}{edgelayer}
		\draw (0.center) to (1);
		\draw (1) to (2);
		\draw [in=90, out=-135] (2) to (4.center);
		\draw [in=90, out=-45] (2) to (3.center);
	\end{pgfonlayer}
\end{tikzpicture}
=
\begin{tikzpicture}
	\begin{pgfonlayer}{nodelayer}
		\node [style=none] (0) at (6.5, 3.25) {};
		\node [style=X] (1) at (6.5, 2.75) {};
		\node [style=scalar] (4) at (6.25, 2.25) {$a$};
		\node [style=scalar] (5) at (6.75, 2.25) {$a$};
		\node [style=none] (6) at (6.75, 1.75) {};
		\node [style=none] (7) at (6.25, 1.75) {};
	\end{pgfonlayer}
	\begin{pgfonlayer}{edgelayer}
		\draw (5) to (6.center);
		\draw (7.center) to (4);
		\draw (1) to (0.center);
		\draw [in=-45, out=90] (5) to (1);
		\draw [in=90, out=-135] (1) to (4);
	\end{pgfonlayer}
\end{tikzpicture}
\hspace*{.5cm}
\begin{tikzpicture}
	\begin{pgfonlayer}{nodelayer}
		\node [style=none] (0) at (5, 2.25) {};
		\node [style=scalar] (1) at (5, 1.75) {$a$};
		\node [style=X] (2) at (5, 1.25) {};
	\end{pgfonlayer}
	\begin{pgfonlayer}{edgelayer}
		\draw (0.center) to (1);
		\draw (1) to (2);
	\end{pgfonlayer}
\end{tikzpicture}
=
\begin{tikzpicture}
	\begin{pgfonlayer}{nodelayer}
		\node [style=none] (0) at (5, 2.25) {};
		\node [style=X] (1) at (5, 1.25) {};
	\end{pgfonlayer}
	\begin{pgfonlayer}{edgelayer}
		\draw (0.center) to (1);
	\end{pgfonlayer}
\end{tikzpicture}
\hspace*{.5cm}
\begin{tikzpicture}
	\begin{pgfonlayer}{nodelayer}
		\node [style=scalar] (28) at (128.25, 2) {$b$};
		\node [style=none] (29) at (128.25, 2.5) {};
		\node [style=scalar] (30) at (128.25, 1.25) {$a$};
		\node [style=none] (31) at (128.25, 0.75) {};
	\end{pgfonlayer}
	\begin{pgfonlayer}{edgelayer}
		\draw (28) to (29.center);
		\draw (31.center) to (30);
		\draw (30) to (28);
	\end{pgfonlayer}
\end{tikzpicture}
=
\begin{tikzpicture}
	\begin{pgfonlayer}{nodelayer}
		\node [style=scalar] (0) at (4, 2) {$ab$};
		\node [style=none] (3) at (4, 2.5) {};
		\node [style=none] (4) at (4, 1.5) {};
	\end{pgfonlayer}
	\begin{pgfonlayer}{edgelayer}
		\draw (0) to (3.center);
		\draw (4.center) to (0);
	\end{pgfonlayer}
\end{tikzpicture}
\hspace*{.5cm}
\begin{tikzpicture}
	\begin{pgfonlayer}{nodelayer}
		\node [style=none] (0) at (5, 2.25) {};
		\node [style=none] (1) at (5, 1.25) {};
		\node [style=scalar] (2) at (5, 1.75) {$1$};
	\end{pgfonlayer}
	\begin{pgfonlayer}{edgelayer}
		\draw (1.center) to (2);
		\draw (2) to (0.center);
	\end{pgfonlayer}
\end{tikzpicture}
=
\begin{tikzpicture}
	\begin{pgfonlayer}{nodelayer}
		\node [style=none] (2) at (5.5, 2.5) {};
		\node [style=none] (3) at (5.5, 1) {};
	\end{pgfonlayer}
	\begin{pgfonlayer}{edgelayer}
		\draw (3.center) to (2.center);
	\end{pgfonlayer}
\end{tikzpicture}
\hspace*{.5cm}
\begin{tikzpicture}
	\begin{pgfonlayer}{nodelayer}
		\node [style=none] (1) at (5.25, 2.75) {};
		\node [style=none] (3) at (5.25, 0.75) {};
		\node [style=Z] (4) at (5.25, 1.25) {};
		\node [style=X] (5) at (5.25, 2.25) {};
		\node [style=scalar] (6) at (5, 1.75) {$a$};
		\node [style=scalar] (7) at (5.5, 1.75) {$b$};
	\end{pgfonlayer}
	\begin{pgfonlayer}{edgelayer}
		\draw (3.center) to (4);
		\draw [in=-90, out=135] (4) to (6);
		\draw [in=-150, out=90] (6) to (5);
		\draw (5) to (1.center);
		\draw [in=90, out=-30] (5) to (7);
		\draw [in=45, out=-90] (7) to (4);
	\end{pgfonlayer}
\end{tikzpicture}
=
\begin{tikzpicture}
	\begin{pgfonlayer}{nodelayer}
		\node [style=none] (1) at (5, 2.25) {};
		\node [style=none] (3) at (5, 1.25) {};
		\node [style=scalar] (6) at (5, 1.75) {$a+b$};
	\end{pgfonlayer}
	\begin{pgfonlayer}{edgelayer}
		\draw (3.center) to (6);
		\draw (6) to (1.center);
	\end{pgfonlayer}
\end{tikzpicture}
\hspace*{.5cm}
\begin{tikzpicture}
	\begin{pgfonlayer}{nodelayer}
		\node [style=none] (1) at (5, 2.25) {};
		\node [style=none] (3) at (5, 1.25) {};
		\node [style=scalar] (6) at (5, 1.75) {$0$};
	\end{pgfonlayer}
	\begin{pgfonlayer}{edgelayer}
		\draw (3.center) to (6);
		\draw (6) to (1.center);
	\end{pgfonlayer}
\end{tikzpicture}
=
\begin{tikzpicture}
	\begin{pgfonlayer}{nodelayer}
		\node [style=X] (7) at (5.5, 2) {};
		\node [style=Z] (8) at (5.5, 1.5) {};
		\node [style=none] (9) at (5.5, 2.5) {};
		\node [style=none] (10) at (5.5, 1) {};
	\end{pgfonlayer}
	\begin{pgfonlayer}{edgelayer}
		\draw (10.center) to (8);
		\draw (7) to (9.center);
	\end{pgfonlayer}
\end{tikzpicture}
$$

\end{definition}

\begin{proposition}[{\cite[Prop.\ 3.9]{ih}}]
Given a ring $k$, $\cb_k$ is a presentation for the prop $\Mat_k$, of matrices over $k$ under the direct sum.
\end{proposition}

One should interpret the grey monoid as addition and the white comonoid as copying.

\begin{definition}[{\cite[Defn.\ 3.42]{ih}}]
Given a field $k$, the prop of {\bf linear relations}, $\LinRel_k$, has morphisms $n\to m$  as linear subspaces of $k^n \oplus k^m$, under relational composition and the direct sum as the tensor product.
\end{definition}

It is only necessary for $k$ to be a principal ideal domain for composition to be well defined, but a field will do for the purposes of this paper.

\begin{definition}[{\cite[Defn.\ 3.44]{ih}}]
Given a field $k$, let $\ih_k$ denote the prop given by the quotient of the coproduct of props $\cb_k^\op+\cb_k$ by the following equations, for all invertible $a \in k$ (where the generators of $\cb_k^\op$ are drawn by reflecting those of $\cb_k$ along the x-axis):
$$
\begin{tabular}{c}
\begin{tikzpicture}
	\begin{pgfonlayer}{nodelayer}
		\node [style=none] (1) at (5, 1.25) {};
		\node [style=scalarop] (2) at (5, 1.75) {$a$};
		\node [style=none] (3) at (5, 2.75) {};
		\node [style=scalar] (5) at (5, 2.25) {$a$};
	\end{pgfonlayer}
	\begin{pgfonlayer}{edgelayer}
		\draw (1.center) to (2);
		\draw (5) to (3.center);
		\draw (2) to (5);
	\end{pgfonlayer}
\end{tikzpicture}
=
\begin{tikzpicture}
	\begin{pgfonlayer}{nodelayer}
		\node [style=none] (6) at (6, 1.25) {};
		\node [style=none] (8) at (6, 2.75) {};
		\node [style=scalarop] (10) at (6, 2.25) {$a$};
		\node [style=scalar] (11) at (6, 1.75) {$a$};
	\end{pgfonlayer}
	\begin{pgfonlayer}{edgelayer}
		\draw (6.center) to (11);
		\draw (11) to (10);
		\draw (10) to (8.center);
	\end{pgfonlayer}
\end{tikzpicture}
=
\begin{tikzpicture}
	\begin{pgfonlayer}{nodelayer}
		\node [style=none] (12) at (7, 2.75) {};
		\node [style=none] (13) at (7, 1.25) {};
	\end{pgfonlayer}
	\begin{pgfonlayer}{edgelayer}
		\draw (13.center) to (12.center);
	\end{pgfonlayer}
\end{tikzpicture}
\hspace*{.1cm}
\begin{tikzpicture}
	\begin{pgfonlayer}{nodelayer}
		\node [style=Z] (0) at (5, 2) {};
		\node [style=Z] (1) at (5.5, 2.5) {};
		\node [style=none] (2) at (5, 1.5) {};
		\node [style=none] (3) at (5.75, 2) {};
		\node [style=none] (4) at (5.75, 1.5) {};
		\node [style=none] (5) at (4.75, 2.5) {};
		\node [style=none] (6) at (5.5, 3) {};
		\node [style=none] (7) at (4.75, 3) {};
	\end{pgfonlayer}
	\begin{pgfonlayer}{edgelayer}
		\draw (7.center) to (5.center);
		\draw [in=120, out=-90] (5.center) to (0);
		\draw (0) to (2.center);
		\draw (0) to (1);
		\draw (1) to (6.center);
		\draw [in=90, out=-60] (1) to (3.center);
		\draw (3.center) to (4.center);
	\end{pgfonlayer}
\end{tikzpicture}
=
\begin{tikzpicture}
	\begin{pgfonlayer}{nodelayer}
		\node [style=Z] (1) at (5.5, 2.5) {};
		\node [style=none] (3) at (5.75, 2) {};
		\node [style=none] (7) at (5.25, 2) {};
		\node [style=Z] (8) at (5.5, 3) {};
		\node [style=none] (9) at (5.75, 3.5) {};
		\node [style=none] (10) at (5.25, 3.5) {};
	\end{pgfonlayer}
	\begin{pgfonlayer}{edgelayer}
		\draw [in=90, out=-60] (1) to (3.center);
		\draw [in=-120, out=90] (7.center) to (1);
		\draw [in=-90, out=60] (8) to (9.center);
		\draw [in=120, out=-90] (10.center) to (8);
		\draw (1) to (8);
	\end{pgfonlayer}
\end{tikzpicture}
=
\begin{tikzpicture}
	\begin{pgfonlayer}{nodelayer}
		\node [style=Z] (0) at (5.5, 2) {};
		\node [style=Z] (1) at (5, 2.5) {};
		\node [style=none] (2) at (5.5, 1.5) {};
		\node [style=none] (3) at (4.75, 2) {};
		\node [style=none] (4) at (4.75, 1.5) {};
		\node [style=none] (5) at (5.75, 2.5) {};
		\node [style=none] (6) at (5, 3) {};
		\node [style=none] (7) at (5.75, 3) {};
	\end{pgfonlayer}
	\begin{pgfonlayer}{edgelayer}
		\draw (7.center) to (5.center);
		\draw [in=60, out=-90] (5.center) to (0);
		\draw (0) to (2.center);
		\draw (0) to (1);
		\draw (1) to (6.center);
		\draw [in=90, out=-120] (1) to (3.center);
		\draw (3.center) to (4.center);
	\end{pgfonlayer}
\end{tikzpicture}
\hspace*{.2cm}
\begin{tikzpicture}
	\begin{pgfonlayer}{nodelayer}
		\node [style=X] (0) at (5, 2) {};
		\node [style=X] (1) at (5.5, 2.5) {};
		\node [style=none] (2) at (5, 1.5) {};
		\node [style=none] (3) at (5.75, 2) {};
		\node [style=none] (4) at (5.75, 1.5) {};
		\node [style=none] (5) at (4.75, 2.5) {};
		\node [style=none] (6) at (5.5, 3) {};
		\node [style=none] (7) at (4.75, 3) {};
	\end{pgfonlayer}
	\begin{pgfonlayer}{edgelayer}
		\draw (7.center) to (5.center);
		\draw [in=120, out=-90] (5.center) to (0);
		\draw (0) to (2.center);
		\draw (0) to (1);
		\draw (1) to (6.center);
		\draw [in=90, out=-60] (1) to (3.center);
		\draw (3.center) to (4.center);
	\end{pgfonlayer}
\end{tikzpicture}
=
\begin{tikzpicture}
	\begin{pgfonlayer}{nodelayer}
		\node [style=X] (1) at (5.5, 2.5) {};
		\node [style=none] (3) at (5.75, 2) {};
		\node [style=none] (7) at (5.25, 2) {};
		\node [style=X] (8) at (5.5, 3) {};
		\node [style=none] (9) at (5.75, 3.5) {};
		\node [style=none] (10) at (5.25, 3.5) {};
	\end{pgfonlayer}
	\begin{pgfonlayer}{edgelayer}
		\draw [in=90, out=-60] (1) to (3.center);
		\draw [in=-120, out=90] (7.center) to (1);
		\draw [in=-90, out=60] (8) to (9.center);
		\draw [in=120, out=-90] (10.center) to (8);
		\draw (1) to (8);
	\end{pgfonlayer}
\end{tikzpicture}
=
\begin{tikzpicture}
	\begin{pgfonlayer}{nodelayer}
		\node [style=X] (0) at (5.5, 2) {};
		\node [style=X] (1) at (5, 2.5) {};
		\node [style=none] (2) at (5.5, 1.5) {};
		\node [style=none] (3) at (4.75, 2) {};
		\node [style=none] (4) at (4.75, 1.5) {};
		\node [style=none] (5) at (5.75, 2.5) {};
		\node [style=none] (6) at (5, 3) {};
		\node [style=none] (7) at (5.75, 3) {};
	\end{pgfonlayer}
	\begin{pgfonlayer}{edgelayer}
		\draw (7.center) to (5.center);
		\draw [in=60, out=-90] (5.center) to (0);
		\draw (0) to (2.center);
		\draw (0) to (1);
		\draw (1) to (6.center);
		\draw [in=90, out=-120] (1) to (3.center);
		\draw (3.center) to (4.center);
	\end{pgfonlayer}
\end{tikzpicture}
\hspace*{.1cm}
\begin{tikzpicture}
	\begin{pgfonlayer}{nodelayer}
		\node [style=Z] (3) at (5.5, 3.25) {};
		\node [style=Z] (6) at (5.5, 4) {};
		\node [style=none] (7) at (5.5, 4.5) {};
		\node [style=none] (8) at (5.5, 2.75) {};
	\end{pgfonlayer}
	\begin{pgfonlayer}{edgelayer}
		\draw (8.center) to (3);
		\draw [in=225, out=135, looseness=1.25] (3) to (6);
		\draw [in=45, out=-45, looseness=1.25] (6) to (3);
		\draw (6) to (7.center);
	\end{pgfonlayer}
\end{tikzpicture}
=
\begin{tikzpicture}
	\begin{pgfonlayer}{nodelayer}
		\node [style=X] (3) at (5.5, 3.25) {};
		\node [style=X] (6) at (5.5, 4) {};
		\node [style=none] (7) at (5.5, 4.5) {};
		\node [style=none] (8) at (5.5, 2.75) {};
	\end{pgfonlayer}
	\begin{pgfonlayer}{edgelayer}
		\draw (8.center) to (3);
		\draw [in=225, out=135, looseness=1.25] (3) to (6);
		\draw [in=45, out=-45, looseness=1.25] (6) to (3);
		\draw (6) to (7.center);
	\end{pgfonlayer}
\end{tikzpicture}
=
\begin{tikzpicture}
	\begin{pgfonlayer}{nodelayer}
		\node [style=none] (7) at (5.5, 4.5) {};
		\node [style=none] (8) at (5.5, 2.75) {};
	\end{pgfonlayer}
	\begin{pgfonlayer}{edgelayer}
		\draw (8.center) to (7.center);
	\end{pgfonlayer}
\end{tikzpicture}
\hspace*{.1cm}
\begin{tikzpicture}
	\begin{pgfonlayer}{nodelayer}
		\node [style=Z] (0) at (6, 3.25) {};
		\node [style=Z] (1) at (6, 2.75) {};
		\node [style=none] (2) at (5.75, 3.75) {};
		\node [style=none] (3) at (6.25, 4.25) {};
		\node [style=none] (4) at (5.75, 4.25) {};
		\node [style=scalar] (5) at (6.25, 3.75) {$-1$};
	\end{pgfonlayer}
	\begin{pgfonlayer}{edgelayer}
		\draw (1) to (0);
		\draw [in=-90, out=30] (0) to (5);
		\draw (5) to (3.center);
		\draw (4.center) to (2.center);
		\draw [in=150, out=-90] (2.center) to (0);
	\end{pgfonlayer}
\end{tikzpicture}
=
\begin{tikzpicture}
	\begin{pgfonlayer}{nodelayer}
		\node [style=X] (0) at (6, 3.25) {};
		\node [style=X] (1) at (6, 2.75) {};
		\node [style=none] (2) at (5.75, 4) {};
		\node [style=none] (3) at (6.25, 4) {};
	\end{pgfonlayer}
	\begin{pgfonlayer}{edgelayer}
		\draw (1) to (0);
		\draw [in=-90, out=30] (0) to (3.center);
		\draw [in=150, out=-90] (2.center) to (0);
	\end{pgfonlayer}
\end{tikzpicture}
\hspace*{.1cm}
\begin{tikzpicture}
	\begin{pgfonlayer}{nodelayer}
		\node [style=X] (4) at (7.5, 3.75) {};
		\node [style=X] (5) at (7.5, 4.25) {};
		\node [style=none] (6) at (7.25, 3.25) {};
		\node [style=none] (7) at (7.75, 2.75) {};
		\node [style=none] (8) at (7.25, 2.75) {};
		\node [style=scalarop] (9) at (7.75, 3.25) {$-1$};
	\end{pgfonlayer}
	\begin{pgfonlayer}{edgelayer}
		\draw (5) to (4);
		\draw [in=90, out=-30] (4) to (9);
		\draw (9) to (7.center);
		\draw (8.center) to (6.center);
		\draw [in=-150, out=90] (6.center) to (4);
	\end{pgfonlayer}
\end{tikzpicture}
=
\begin{tikzpicture}
	\begin{pgfonlayer}{nodelayer}
		\node [style=Z] (0) at (6, 3.5) {};
		\node [style=Z] (1) at (6, 4) {};
		\node [style=none] (2) at (5.75, 2.75) {};
		\node [style=none] (3) at (6.25, 2.75) {};
	\end{pgfonlayer}
	\begin{pgfonlayer}{edgelayer}
		\draw (1) to (0);
		\draw [in=90, out=-30] (0) to (3.center);
		\draw [in=-150, out=90] (2.center) to (0);
	\end{pgfonlayer}
\end{tikzpicture}
\end{tabular}
$$
\end{definition}

\begin{theorem}[{\cite[Thm.\ 3.49]{ih}}]
$\ih_k$ is a presentation for $\LinRel_k$.
\end{theorem}

There is an interesting folklore result which was remarked in \cite{cole}: \footnote{One should note that the black box is the antipode, {\em not} the Fourier transform/Hadamard gate.}:
\begin{lemma}
\label{lemma:phasefree}
For a prime number $p$, $\ih_{\F_p}$ is a presentation for the phase-free, Fourier-free $p$-dimensional qudit ZX-calculus, modulo invertible scalars.
\end{lemma}

The following theorem will be useful for graphical manipulations:
\begin{theorem}[\cite{spider} (Spider Theorem)]
All connected components of special commutative Frobenius algebras with the same arity are equal.
\end{theorem}

That is to say, we can unambiguously refer to these connected components by spiders.  In $\ih_k$, there are two spiders, so for example we can apply spider fusion to the following circuit:\hspace*{.2cm}
$
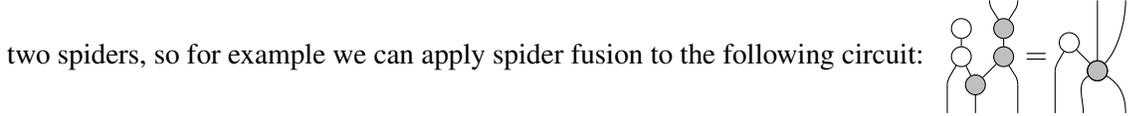
\begin{tikzpicture}
	\begin{pgfonlayer}{nodelayer}
		\node [style=Z] (0) at (6, 4) {};
		\node [style=Z] (1) at (6, 3.5) {};
		\node [style=X] (2) at (6.75, 3.5) {};
		\node [style=X] (3) at (6.75, 4) {};
		\node [style=none] (4) at (6.5, 4.5) {};
		\node [style=none] (5) at (7, 4.5) {};
		\node [style=X] (6) at (6.25, 3) {};
		\node [style=none] (7) at (5.75, 3) {};
		\node [style=none] (8) at (6.25, 2.5) {};
		\node [style=none] (9) at (7, 3) {};
		\node [style=none] (10) at (7, 2.5) {};
		\node [style=none] (11) at (5.75, 2.5) {};
	\end{pgfonlayer}
	\begin{pgfonlayer}{edgelayer}
		\draw [in=240, out=90, looseness=0.75] (7.center) to (1);
		\draw (1) to (6);
		\draw (6) to (8.center);
		\draw (6) to (2);
		\draw (2) to (3);
		\draw [in=-90, out=60] (3) to (5.center);
		\draw [in=-90, out=120] (3) to (4.center);
		\draw (1) to (0);
		\draw (11.center) to (7.center);
		\draw (10.center) to (9.center);
		\draw [in=300, out=90] (9.center) to (2);
	\end{pgfonlayer}
\end{tikzpicture}
=
\begin{tikzpicture}
	\begin{pgfonlayer}{nodelayer}
		\node [style=Z] (13) at (8.25, 3.75) {};
		\node [style=X] (14) at (8.75, 3.25) {};
		\node [style=X] (15) at (8.75, 3.25) {};
		\node [style=none] (16) at (8.75, 4.5) {};
		\node [style=none] (17) at (9.25, 4.5) {};
		\node [style=X] (18) at (8.75, 3.25) {};
		\node [style=none] (19) at (8, 3.25) {};
		\node [style=none] (20) at (8.5, 2.5) {};
		\node [style=none] (21) at (9.25, 2.5) {};
		\node [style=none] (23) at (8, 2.5) {};
	\end{pgfonlayer}
	\begin{pgfonlayer}{edgelayer}
		\draw [in=240, out=90, looseness=0.75] (19.center) to (13);
		\draw (13) to (18);
		\draw [in=90, out=-150] (18) to (20.center);
		\draw [in=-90, out=45, looseness=0.75] (15) to (17.center);
		\draw [in=-90, out=90, looseness=0.50] (15) to (16.center);
		\draw (23.center) to (19.center);
		\draw [in=-30, out=90] (21.center) to (14);
	\end{pgfonlayer}
\end{tikzpicture}
$

We  shall use the following results:
\begin{lemma}[{\cite[(D4)]{ih}}, {\cite[p. 4]{control}}, {\cite[(D3)]{ih}}]
$$
\begin{tikzpicture}
	\begin{pgfonlayer}{nodelayer}
		\node [style=Z] (0) at (6.75, 4.5) {};
		\node [style=X] (1) at (7.75, 3.5) {};
		\node [style=none] (3) at (6.75, 5) {};
		\node [style=none] (4) at (8, 5) {};
		\node [style=none] (5) at (6.5, 3) {};
		\node [style=none] (6) at (7.75, 3) {};
	\end{pgfonlayer}
	\begin{pgfonlayer}{edgelayer}
		\draw (3.center) to (0);
		\draw [in=60, out=-90] (4.center) to (1);
		\draw (1) to (6.center);
		\draw [in=90, out=-120] (0) to (5.center);
		\draw (1) to (0);
	\end{pgfonlayer}
\end{tikzpicture}
=
\begin{tikzpicture}
	\begin{pgfonlayer}{nodelayer}
		\node [style=Z] (0) at (6.75, 3.5) {};
		\node [style=X] (1) at (7.75, 4.5) {};
		\node [style=scalar] (2) at (7.25, 4) {$-1$};
		\node [style=none] (3) at (6.75, 3) {};
		\node [style=none] (4) at (8, 3) {};
		\node [style=none] (5) at (6.5, 5) {};
		\node [style=none] (6) at (7.75, 5) {};
	\end{pgfonlayer}
	\begin{pgfonlayer}{edgelayer}
		\draw (3.center) to (0);
		\draw (0) to (2);
		\draw (2) to (1);
		\draw [in=-60, out=90] (4.center) to (1);
		\draw (1) to (6.center);
		\draw [in=-90, out=120] (0) to (5.center);
	\end{pgfonlayer}
\end{tikzpicture}
\hspace*{1cm}
\begin{tikzpicture}
	\begin{pgfonlayer}{nodelayer}
		\node [style=scalar] (2) at (7.75, 3.25) {$a$};
		\node [style=Z] (3) at (8, 3.75) {};
		\node [style=Z] (4) at (7.5, 2.75) {};
		\node [style=none] (5) at (8.25, 3.25) {};
		\node [style=none] (6) at (8.25, 2.5) {};
		\node [style=none] (7) at (7.25, 3.25) {};
		\node [style=none] (8) at (7.25, 4) {};
	\end{pgfonlayer}
	\begin{pgfonlayer}{edgelayer}
		\draw (8.center) to (7.center);
		\draw [in=120, out=-90, looseness=0.75] (7.center) to (4);
		\draw [in=-90, out=45] (4) to (2);
		\draw [in=-135, out=90] (2) to (3);
		\draw [in=90, out=-60, looseness=0.75] (3) to (5.center);
		\draw (5.center) to (6.center);
	\end{pgfonlayer}
\end{tikzpicture}
=
\begin{tikzpicture}
	\begin{pgfonlayer}{nodelayer}
		\node [style=scalar] (2) at (7.75, 3.25) {$a$};
		\node [style=X] (3) at (8, 3.75) {};
		\node [style=X] (4) at (7.5, 2.75) {};
		\node [style=none] (5) at (8.25, 3.25) {};
		\node [style=none] (6) at (8.25, 2.5) {};
		\node [style=none] (7) at (7.25, 3.25) {};
		\node [style=none] (8) at (7.25, 4) {};
	\end{pgfonlayer}
	\begin{pgfonlayer}{edgelayer}
		\draw (8.center) to (7.center);
		\draw [in=120, out=-90, looseness=0.75] (7.center) to (4);
		\draw [in=-90, out=45] (4) to (2);
		\draw [in=-135, out=90] (2) to (3);
		\draw [in=90, out=-60, looseness=0.75] (3) to (5.center);
		\draw (5.center) to (6.center);
	\end{pgfonlayer}
\end{tikzpicture}
=
\begin{tikzpicture}
	\begin{pgfonlayer}{nodelayer}
		\node [style=scalarop] (2) at (7.75, 3.25) {$a$};
		\node [style=none] (6) at (7.75, 2.5) {};
		\node [style=none] (8) at (7.75, 4) {};
	\end{pgfonlayer}
	\begin{pgfonlayer}{edgelayer}
		\draw (6.center) to (2);
		\draw (2) to (8.center);
	\end{pgfonlayer}
\end{tikzpicture}
\hspace*{1cm}
\begin{tikzpicture}
	\begin{pgfonlayer}{nodelayer}
		\node [style=none] (2) at (7.75, 3.75) {};
		\node [style=none] (3) at (7.75, 2.75) {};
		\node [style=scalarop] (5) at (7.75, 3.25) {$-1$};
	\end{pgfonlayer}
	\begin{pgfonlayer}{edgelayer}
		\draw (5) to (3.center);
		\draw (5) to (2.center);
	\end{pgfonlayer}
\end{tikzpicture}
=
\begin{tikzpicture}
	\begin{pgfonlayer}{nodelayer}
		\node [style=none] (2) at (7.75, 3.75) {};
		\node [style=none] (3) at (7.75, 2.75) {};
		\node [style=scalar] (5) at (7.75, 3.25) {$-1$};
	\end{pgfonlayer}
	\begin{pgfonlayer}{edgelayer}
		\draw (5) to (3.center);
		\draw (5) to (2.center);
	\end{pgfonlayer}
\end{tikzpicture}
=
\begin{tikzpicture}
	\begin{pgfonlayer}{nodelayer}
		\node [style=Z] (0) at (7, 3.5) {};
		\node [style=X] (2) at (7.25, 4) {};
		\node [style=none] (4) at (7.5, 3.5) {};
		\node [style=none] (5) at (7.5, 3.25) {};
		\node [style=none] (6) at (6.75, 4) {};
		\node [style=none] (7) at (6.75, 4.25) {};
	\end{pgfonlayer}
	\begin{pgfonlayer}{edgelayer}
		\draw (2) to (0);
		\draw [in=120, out=-90, looseness=0.75] (6.center) to (0);
		\draw [in=300, out=90, looseness=0.75] (4.center) to (2);
		\draw (5.center) to (4.center);
		\draw (6.center) to (7.center);
	\end{pgfonlayer}
\end{tikzpicture}
=
\begin{tikzpicture}
	\begin{pgfonlayer}{nodelayer}
		\node [style=Z] (0) at (7, 4) {};
		\node [style=X] (2) at (7.25, 3.5) {};
		\node [style=none] (4) at (7.5, 4) {};
		\node [style=none] (5) at (7.5, 4.25) {};
		\node [style=none] (6) at (6.75, 3.5) {};
		\node [style=none] (7) at (6.75, 3.25) {};
	\end{pgfonlayer}
	\begin{pgfonlayer}{edgelayer}
		\draw (2) to (0);
		\draw [in=-120, out=90, looseness=0.75] (6.center) to (0);
		\draw [in=-300, out=-90, looseness=0.75] (4.center) to (2);
		\draw (5.center) to (4.center);
		\draw (6.center) to (7.center);
	\end{pgfonlayer}
\end{tikzpicture}
$$
\end{lemma}

Because of the symmetry of $-1$, we use the following (symmetric) notation for the antipode:
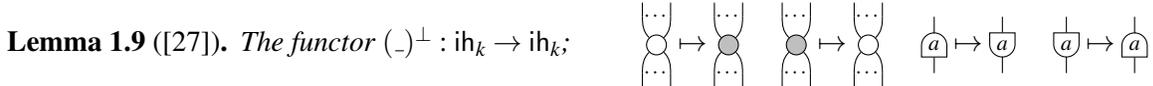
\hfil
$
\begin{tikzpicture}
	\begin{pgfonlayer}{nodelayer}
		\node [style=none] (2) at (7.75, 3.75) {};
		\node [style=none] (3) at (7.75, 2.75) {};
		\node [style=s] (5) at (7.75, 3.25) {};
	\end{pgfonlayer}
	\begin{pgfonlayer}{edgelayer}
		\draw (5) to (3.center);
		\draw (5) to (2.center);
	\end{pgfonlayer}
\end{tikzpicture}
:=
\begin{tikzpicture}
	\begin{pgfonlayer}{nodelayer}
		\node [style=none] (2) at (7.75, 3.75) {};
		\node [style=none] (3) at (7.75, 2.75) {};
		\node [style=scalar] (5) at (7.75, 3.25) {$-1$};
	\end{pgfonlayer}
	\begin{pgfonlayer}{edgelayer}
		\draw (5) to (3.center);
		\draw (5) to (2.center);
	\end{pgfonlayer}
\end{tikzpicture}
$

\begin{lemma}[{\cite{ortho}}]
The functor $(\_)^\perp:\ih_k\to \ih_k$;
\hfil
$
\begin{tikzpicture}
	\begin{pgfonlayer}{nodelayer}
		\node [style=Z] (0) at (7, 3.5) {};
		\node [style=none] (1) at (6.75, 4) {};
		\node [style=none] (2) at (7.25, 4) {};
		\node [style=none] (3) at (6.75, 3) {};
		\node [style=none] (4) at (7.25, 3) {};
		\node [style=none] (5) at (6.75, 4.25) {};
		\node [style=none] (6) at (7.25, 4.25) {};
		\node [style=none] (7) at (7.25, 2.75) {};
		\node [style=none] (8) at (6.75, 2.75) {};
		\node [style=none] (9) at (7, 4) {$\cdots$};
		\node [style=none] (10) at (7, 3) {$\cdots$};
	\end{pgfonlayer}
	\begin{pgfonlayer}{edgelayer}
		\draw [in=-90, out=135] (0) to (1.center);
		\draw [in=-90, out=45] (0) to (2.center);
		\draw [in=-45, out=90] (4.center) to (0);
		\draw [in=90, out=-135] (0) to (3.center);
		\draw (8.center) to (3.center);
		\draw (7.center) to (4.center);
		\draw (2.center) to (6.center);
		\draw (5.center) to (1.center);
	\end{pgfonlayer}
\end{tikzpicture}
\mapsto
\begin{tikzpicture}
	\begin{pgfonlayer}{nodelayer}
		\node [style=X] (0) at (7, 3.5) {};
		\node [style=none] (1) at (6.75, 4) {};
		\node [style=none] (2) at (7.25, 4) {};
		\node [style=none] (3) at (6.75, 3) {};
		\node [style=none] (4) at (7.25, 3) {};
		\node [style=none] (5) at (6.75, 4.25) {};
		\node [style=none] (6) at (7.25, 4.25) {};
		\node [style=none] (7) at (7.25, 2.75) {};
		\node [style=none] (8) at (6.75, 2.75) {};
		\node [style=none] (9) at (7, 4) {$\cdots$};
		\node [style=none] (10) at (7, 3) {$\cdots$};
	\end{pgfonlayer}
	\begin{pgfonlayer}{edgelayer}
		\draw [in=-90, out=135] (0) to (1.center);
		\draw [in=-90, out=45] (0) to (2.center);
		\draw [in=-45, out=90] (4.center) to (0);
		\draw [in=90, out=-135] (0) to (3.center);
		\draw (8.center) to (3.center);
		\draw (7.center) to (4.center);
		\draw (2.center) to (6.center);
		\draw (5.center) to (1.center);
	\end{pgfonlayer}
\end{tikzpicture}
\hspace*{.5cm}
\begin{tikzpicture}
	\begin{pgfonlayer}{nodelayer}
		\node [style=X] (0) at (7, 3.5) {};
		\node [style=none] (1) at (6.75, 4) {};
		\node [style=none] (2) at (7.25, 4) {};
		\node [style=none] (3) at (6.75, 3) {};
		\node [style=none] (4) at (7.25, 3) {};
		\node [style=none] (5) at (6.75, 4.25) {};
		\node [style=none] (6) at (7.25, 4.25) {};
		\node [style=none] (7) at (7.25, 2.75) {};
		\node [style=none] (8) at (6.75, 2.75) {};
		\node [style=none] (9) at (7, 4) {$\cdots$};
		\node [style=none] (10) at (7, 3) {$\cdots$};
	\end{pgfonlayer}
	\begin{pgfonlayer}{edgelayer}
		\draw [in=-90, out=135] (0) to (1.center);
		\draw [in=-90, out=45] (0) to (2.center);
		\draw [in=-45, out=90] (4.center) to (0);
		\draw [in=90, out=-135] (0) to (3.center);
		\draw (8.center) to (3.center);
		\draw (7.center) to (4.center);
		\draw (2.center) to (6.center);
		\draw (5.center) to (1.center);
	\end{pgfonlayer}
\end{tikzpicture}
\mapsto
\begin{tikzpicture}
	\begin{pgfonlayer}{nodelayer}
		\node [style=Z] (0) at (7, 3.5) {};
		\node [style=none] (1) at (6.75, 4) {};
		\node [style=none] (2) at (7.25, 4) {};
		\node [style=none] (3) at (6.75, 3) {};
		\node [style=none] (4) at (7.25, 3) {};
		\node [style=none] (5) at (6.75, 4.25) {};
		\node [style=none] (6) at (7.25, 4.25) {};
		\node [style=none] (7) at (7.25, 2.75) {};
		\node [style=none] (8) at (6.75, 2.75) {};
		\node [style=none] (9) at (7, 4) {$\cdots$};
		\node [style=none] (10) at (7, 3) {$\cdots$};
	\end{pgfonlayer}
	\begin{pgfonlayer}{edgelayer}
		\draw [in=-90, out=135] (0) to (1.center);
		\draw [in=-90, out=45] (0) to (2.center);
		\draw [in=-45, out=90] (4.center) to (0);
		\draw [in=90, out=-135] (0) to (3.center);
		\draw (8.center) to (3.center);
		\draw (7.center) to (4.center);
		\draw (2.center) to (6.center);
		\draw (5.center) to (1.center);
	\end{pgfonlayer}
\end{tikzpicture}
\hspace*{.5cm}
\begin{tikzpicture}
	\begin{pgfonlayer}{nodelayer}
		\node [style=none] (2) at (7.75, 3.75) {};
		\node [style=none] (3) at (7.75, 2.75) {};
		\node [style=scalar] (5) at (7.75, 3.25) {$a$};
	\end{pgfonlayer}
	\begin{pgfonlayer}{edgelayer}
		\draw (5) to (3.center);
		\draw (5) to (2.center);
	\end{pgfonlayer}
\end{tikzpicture}
\mapsto
\begin{tikzpicture}
	\begin{pgfonlayer}{nodelayer}
		\node [style=none] (2) at (7.75, 3.75) {};
		\node [style=none] (3) at (7.75, 2.75) {};
		\node [style=scalarop] (5) at (7.75, 3.25) {$a$};
	\end{pgfonlayer}
	\begin{pgfonlayer}{edgelayer}
		\draw (5) to (3.center);
		\draw (5) to (2.center);
	\end{pgfonlayer}
\end{tikzpicture}
\hspace*{.5cm}
\begin{tikzpicture}
	\begin{pgfonlayer}{nodelayer}
		\node [style=none] (2) at (7.75, 3.75) {};
		\node [style=none] (3) at (7.75, 2.75) {};
		\node [style=scalarop] (5) at (7.75, 3.25) {$a$};
	\end{pgfonlayer}
	\begin{pgfonlayer}{edgelayer}
		\draw (5) to (3.center);
		\draw (5) to (2.center);
	\end{pgfonlayer}
\end{tikzpicture}
\mapsto
\begin{tikzpicture}
	\begin{pgfonlayer}{nodelayer}
		\node [style=none] (2) at (7.75, 3.75) {};
		\node [style=none] (3) at (7.75, 2.75) {};
		\node [style=scalar] (5) at (7.75, 3.25) {$a$};
	\end{pgfonlayer}
	\begin{pgfonlayer}{edgelayer}
		\draw (5) to (3.center);
		\draw (5) to (2.center);
	\end{pgfonlayer}
\end{tikzpicture}
$

is the isomorphism which takes linear subspaces to their orthogonal complement, that is to say:
$$
V \mapsto V^\perp := \{  v \in V: \forall w \in V , \langle v,w\rangle = 0\}
$$

\end{lemma}

Notice that the orthogonal complement is an involution so that $(V^\perp)^\perp = V$.

\section{Lagrangian relations}
\label{sec:sym}

Now that we have a graphical presentation of linear relations, we can do the same for (linear) Lagrangian relations.  We first recall some of the basic theory of symplectic vector spaces.  This is expounded upon in much greater generality in the not-necessarily-linear case in \cite{weinstein}.  In this entire paper, we only care about the linear and affine cases; and things will be assumed to be linear unless otherwise stated.  As previously mentioned, Lagrangian relations (and their affine counterpart) have been studied within the context of monoidal categories  to model electrical circuits among other things \cite{passive,network,coya}, although, to the knowledge of the authors, no proof of universality exists in the literature.

\begin{definition}
  Given a field  $k$ and a $k$-vector space $V$, a {\bf symplectic form} on $V$ is a bilinear map $\omega:V\times V\to k$ which is:
\begin{itemize}
 \item {\bf Alternating:} $\forall v \in V$, $\omega(v,v)=0$
 \item {\bf Non-degenerate:} If $\forall w \in V \omega(v,w)=0 $ then $ v=0$.
\end{itemize}
  A {\bf symplectic vector space} is a vector space equipped with a symplectic form. A (linear) {\bf symplectomorphism} is a linear isomorphism between symplectic vector spaces that preserves the symplectic form.
\end{definition}

\begin{lemma}
\label{lemma:sform}
Every vector space $k^{2n}$ is equipped with a bilinear form given by the following block matrix:
$$
\omega:=
\begin{bmatrix}
0_n & I_n\\
-I_n & 0_n
\end{bmatrix}
$$
so that $\omega(v,w) := v \omega w^T$.
Moreover, every finite-dimensional symplectic vector space over $k$ is symplectomorphic to one of the form $k^{2n}$ with such a symplectic form.
\end{lemma}

\begin{definition}

Let $W \subseteq V$ be a linear subspace of a symplectic space $V$.
The {\bf symplectic dual} of the subspace $W$ is defined to be
$
W^\omega:= \{v \in V : \forall w \in W, \omega(v,w)=0 \}
$.
A linear subspace  $W$ of a symplectic vector space $V$ is {\bf isotropic} when $W^\omega \supseteq W$, {\bf coisotropic} when $W^\omega \subseteq W$ and {\bf Lagrangian} when $W^\omega=W$.

\end{definition}

\begin{lemma}
Every symplectomorphism $f:V\to V$ induces a Lagrangian relation $\Gamma_f:=\{ (fv, v) | v \in V \}$.
\end{lemma}

These spaces have a natural grading into two distinct parts $V \oplus W \subseteq k^n \oplus k^n$. By analogy to the case of quantum stabilizer theory, we call the left part the \textit{X-grading} and the right part the \textit{Z-grading}.

As a matter of convention, we consider linear subspaces as being represented as the row space of a matrix. So in particular, a symplectic subspace of $k^{2n}$ is represented by a matrix of the form $[X|Z]$ where $X,Z$ are both $n\times n$-dimensional matrices.
An isotropic subspace can equivalently be described as a matrix $[X|Z]$ so that $[X|Z] \omega [X|Z]^T = 0$.
Moreover, a Lagrangian subspace can be described as a matrix as above which additionally has rank $n$.

\begin{definition}
Given a field $k$, the prop of {\bf Lagrangian relations},  $\Lag\Rel_k$, has morphisms $k^{2n}\to k^{2m}$ as Lagrangian subspaces of the symplectic vector space $k^{n+m} \oplus k^{n+m}$ with symplectic form given above.  Composition is given by relational composition and the tensor product is given by the direct sum.
\end{definition}

The direct sum of Lagrangian subspaces is graphically depicted as follows:
$$
\begin{tikzpicture}
	\begin{pgfonlayer}{nodelayer}
		\node [style=map] (616) at (272, 0) {$V$};
		\node [style=none] (617) at (271.75, 1) {};
		\node [style=none] (618) at (272.25, 1) {};
	\end{pgfonlayer}
	\begin{pgfonlayer}{edgelayer}
		\draw [in=-90, out=60] (616) to (618.center);
		\draw [in=-90, out=120] (616) to (617.center);
	\end{pgfonlayer}
\end{tikzpicture}
\oplus
\begin{tikzpicture}
	\begin{pgfonlayer}{nodelayer}
		\node [style=map] (616) at (272, 0) {$W$};
		\node [style=none] (617) at (271.75, 1) {};
		\node [style=none] (618) at (272.25, 1) {};
	\end{pgfonlayer}
	\begin{pgfonlayer}{edgelayer}
		\draw [in=-90, out=60] (616) to (618.center);
		\draw [in=-90, out=120] (616) to (617.center);
	\end{pgfonlayer}
\end{tikzpicture}
:=
\begin{tikzpicture}
	\begin{pgfonlayer}{nodelayer}
		\node [style=map] (616) at (272, 0) {$V$};
		\node [style=none] (617) at (271.75, 1) {};
		\node [style=none] (618) at (272.75, 1) {};
		\node [style=map] (619) at (272.75, 0) {$W$};
		\node [style=none] (620) at (272, 1) {};
		\node [style=none] (621) at (273, 1) {};
	\end{pgfonlayer}
	\begin{pgfonlayer}{edgelayer}
		\draw [in=-90, out=60] (616) to (618.center);
		\draw [in=-90, out=120] (616) to (617.center);
		\draw [in=-90, out=60] (619) to (621.center);
		\draw [in=-90, out=120] (619) to (620.center);
	\end{pgfonlayer}
\end{tikzpicture}
$$
where we are grouping the $X$ gradings together on the left and the $Z$ gradings together on the right. Note that this means the embedding of $\Lag\Rel_k$ into $\LinRel_k$ preserves the monoidal product only up to isomorphism. More precisely, we have the following fact.

\begin{lemma}
\label{lemma:strong}
The forgetful functor $E:\Lag\Rel_k \to \LinRel_k$  is faithful, strong symmetric monoidal.
\end{lemma}

\begin{proof}
  Functoriality and faithfulness is immediate. The strong monoidal structure is given by $E(I) = I$ and
  \[ E(A) \oplus E(B) := A \oplus A \oplus B \oplus B \xrightarrow{1 \oplus \sigma \oplus 1} A \oplus B \oplus A \oplus B =: E(A \oplus B). \]
  The symmetric monoidal structure on $\Lag\Rel_k$ is chosen such that it is consistent with the monoidal structure above.
\end{proof}

Due to the above lemma, we will regard $\Lag\Rel_k$ as a symmetric monoidal subcategory of $\LinRel_k$.
As such, we can ask what the generators of $\Lag\Rel_k$ look like in terms of string diagrams of $\ih_k$ generators. We first describe what it means to be a Lagrangian relation in pictures, where the $X$ block is the wire on the left and $Z$ block is the wire on the right:
\begin{equation}
\label{eq:lag}
\begin{tikzpicture}
	\begin{pgfonlayer}{nodelayer}
		\node [style=map] (0) at (0.75, -1) {$W$};
		\node [style=none] (1) at (0.5, 0) {};
		\node [style=none] (2) at (1, 0) {};
	\end{pgfonlayer}
	\begin{pgfonlayer}{edgelayer}
		\draw [in=120, out=-90] (1.center) to (0);
		\draw [in=-90, out=60] (0) to (2.center);
	\end{pgfonlayer}
\end{tikzpicture}
=
\begin{tikzpicture}
	\begin{pgfonlayer}{nodelayer}
		\node [style=map] (0) at (0.75, -1.75) {$W^\perp$};
		\node [style=none] (1) at (0.5, -1) {};
		\node [style=none] (2) at (1, -1) {};
		\node [style=none] (3) at (1, 0) {};
		\node [style=none] (4) at (0.5, 0) {};
		\node [style=s] (5) at (1, -1) {};
	\end{pgfonlayer}
	\begin{pgfonlayer}{edgelayer}
		\draw [in=120, out=-90] (1.center) to (0);
		\draw [in=-90, out=60] (0) to (2.center);
		\draw [in=-90, out=90] (2.center) to (4.center);
		\draw [in=-270, out=-90] (3.center) to (1.center);
	\end{pgfonlayer}
\end{tikzpicture}
\end{equation}
Algebraically, for $W$ a subspace of $V$, the right hand side is interpreted as follows:
\begin{align*}
W^\omega :&= \{(v_1,v_2) \in V : \forall (w_1,w_2) \in W, \omega((v_1,v_2),(w_1,w_2))=0 \}\\
                    &= \{(v_1,v_2) \in V : \forall (w_1,w_2) \in W,  \langle (v_2,-v_1) ,(w_1,w_2)\rangle =0 \}\\
                    &= \{(v_2,-v_1) \in V : \forall (w_1,w_2) \in W,  \langle (v_1,v_2) ,(w_1,w_2)\rangle =0 \}
\end{align*}

The category of Lagrangian relations is compact closed.  Given a relation $V$ between symplectic vector spaces, we can curry it into a state $\hat V$; and similarily, we can uncurry a state $W$ into a process $\widecheck W$,
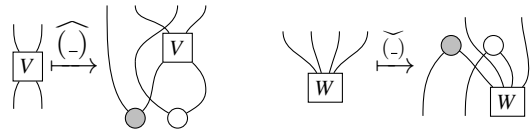
$$
\begin{tikzpicture}
	\begin{pgfonlayer}{nodelayer}
		\node [style=map] (0) at (0.75, -1.75) {$V$};
		\node [style=none] (1) at (0.5, -1) {};
		\node [style=none] (2) at (1, -1) {};
		\node [style=none] (3) at (0.5, -2.5) {};
		\node [style=none] (4) at (1, -2.5) {};
	\end{pgfonlayer}
	\begin{pgfonlayer}{edgelayer}
		\draw [in=120, out=-90] (1.center) to (0);
		\draw [in=-90, out=60] (0) to (2.center);
		\draw [in=-60, out=90] (4.center) to (0);
		\draw [in=90, out=-120] (0) to (3.center);
	\end{pgfonlayer}
\end{tikzpicture}
\xmapsto{\hat{(\_)} }
\begin{tikzpicture}
	\begin{pgfonlayer}{nodelayer}
		\node [style=map] (0) at (0.75, -1.75) {$V$};
		\node [style=none] (1) at (0, -1) {};
		\node [style=none] (2) at (1.25, -1) {};
		\node [style=none] (4) at (1.25, -2.5) {};
		\node [style=X] (5) at (0, -3) {};
		\node [style=Z] (6) at (0.75, -3) {};
		\node [style=none] (7) at (0.75, -1) {};
		\node [style=none] (8) at (-0.5, -1) {};
		\node [style=none] (9) at (0, -2) {};
	\end{pgfonlayer}
	\begin{pgfonlayer}{edgelayer}
		\draw [in=120, out=-90] (1.center) to (0);
		\draw [in=-90, out=60] (0) to (2.center);
		\draw [in=-45, out=90] (4.center) to (0);
		\draw [in=-90, out=135, looseness=0.75] (5) to (8.center);
		\draw [in=30, out=-90] (4.center) to (6);
		\draw [in=90, out=-90] (7.center) to (9.center);
		\draw [in=150, out=-90] (9.center) to (6);
		\draw [in=45, out=-135] (0) to (5);
	\end{pgfonlayer}
\end{tikzpicture}
\hspace*{1cm}
\begin{tikzpicture}
	\begin{pgfonlayer}{nodelayer}
		\node [style=map] (0) at (1.5, -2) {$W$};
		\node [style=none] (1) at (1.25, -1) {};
		\node [style=none] (2) at (2.25, -1) {};
		\node [style=none] (6) at (1.75, -1) {};
		\node [style=none] (7) at (0.75, -1) {};
	\end{pgfonlayer}
	\begin{pgfonlayer}{edgelayer}
		\draw [in=105, out=-90] (1.center) to (0);
		\draw [in=-90, out=60] (0) to (2.center);
		\draw [in=120, out=-90] (7.center) to (0);
		\draw [in=-90, out=75] (0) to (6.center);
	\end{pgfonlayer}
\end{tikzpicture}
\xmapsto{\widecheck{(\_)} }
\begin{tikzpicture}
	\begin{pgfonlayer}{nodelayer}
		\node [style=map] (0) at (1.75, -2.25) {$W$};
		\node [style=none] (1) at (1, -0.75) {};
		\node [style=none] (2) at (2, -0.75) {};
		\node [style=none] (6) at (1.5, -1.25) {};
		\node [style=none] (7) at (0.75, -1.25) {};
		\node [style=Z] (8) at (1.5, -1.25) {};
		\node [style=X] (9) at (0.75, -1.25) {};
		\node [style=none] (10) at (0.25, -2.5) {};
		\node [style=none] (11) at (1, -2.5) {};
	\end{pgfonlayer}
	\begin{pgfonlayer}{edgelayer}
		\draw [in=105, out=-90, looseness=1.25] (1.center) to (0);
		\draw [in=-90, out=45, looseness=0.75] (0) to (2.center);
		\draw [in=135, out=-30] (7.center) to (0);
		\draw [in=-45, out=75] (0) to (6.center);
		\draw [in=-135, out=90] (10.center) to (9);
		\draw [in=-135, out=90] (11.center) to (8);
	\end{pgfonlayer}
\end{tikzpicture}
$$
It is easy to see that these two constructions are inverse to each other.
This allows us to derive a graphical criteria for abitrary Lagrangian relations, generalizing Equation \ref{eq:lag}:
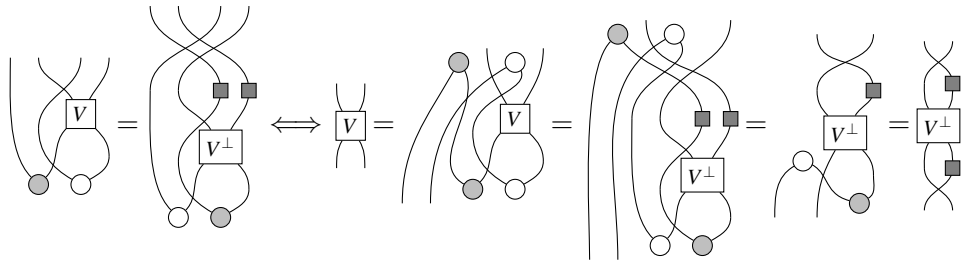
$$
\begin{tikzpicture}
	\begin{pgfonlayer}{nodelayer}
		\node [style=map] (0) at (0.75, -1.75) {$V$};
		\node [style=none] (1) at (0, -0.75) {};
		\node [style=none] (2) at (1.25, -0.75) {};
		\node [style=none] (4) at (1.25, -2.5) {};
		\node [style=X] (5) at (0, -3) {};
		\node [style=Z] (6) at (0.75, -3) {};
		\node [style=none] (7) at (0.75, -0.75) {};
		\node [style=none] (8) at (-0.5, -0.75) {};
		\node [style=none] (9) at (0, -2) {};
	\end{pgfonlayer}
	\begin{pgfonlayer}{edgelayer}
		\draw [in=120, out=-90] (1.center) to (0);
		\draw [in=-90, out=60] (0) to (2.center);
		\draw [in=-45, out=90] (4.center) to (0);
		\draw [in=-90, out=135, looseness=0.75] (5) to (8.center);
		\draw [in=30, out=-90] (4.center) to (6);
		\draw [in=90, out=-90] (7.center) to (9.center);
		\draw [in=150, out=-90] (9.center) to (6);
		\draw [in=45, out=-135] (0) to (5);
	\end{pgfonlayer}
\end{tikzpicture}
=
\begin{tikzpicture}
	\begin{pgfonlayer}{nodelayer}
		\node [style=map] (43) at (13.5, -1.75) {$V^\perp$};
		\node [style=none] (44) at (12.75, -0.75) {};
		\node [style=none] (45) at (14, -0.75) {};
		\node [style=none] (46) at (14, -2.5) {};
		\node [style=Z] (47) at (12.75, -3) {};
		\node [style=X] (48) at (13.5, -3) {};
		\node [style=none] (49) at (13.5, -0.75) {};
		\node [style=none] (50) at (12.25, -0.75) {};
		\node [style=none] (51) at (12.75, -2) {};
		\node [style=none] (52) at (13.5, 0.75) {};
		\node [style=none] (53) at (14, 0.75) {};
		\node [style=none] (54) at (12.25, 0.75) {};
		\node [style=none] (55) at (12.75, 0.75) {};
		\node [style=s] (56) at (14, -0.75) {};
		\node [style=s] (57) at (13.5, -0.75) {};
		\node [style=none] (58) at (13.25, -3.5) {};
	\end{pgfonlayer}
	\begin{pgfonlayer}{edgelayer}
		\draw [in=120, out=-90] (44.center) to (43);
		\draw [in=-90, out=60] (43) to (45.center);
		\draw [in=-45, out=90] (46.center) to (43);
		\draw [in=-90, out=135, looseness=0.75] (47) to (50.center);
		\draw [in=30, out=-90] (46.center) to (48);
		\draw [in=90, out=-90] (49.center) to (51.center);
		\draw [in=150, out=-90] (51.center) to (48);
		\draw [in=45, out=-135] (43) to (47);
		\draw [in=270, out=90] (45.center) to (55.center);
		\draw [in=270, out=90] (49.center) to (54.center);
		\draw [in=270, out=90] (44.center) to (53.center);
		\draw [in=270, out=90] (50.center) to (52.center);
	\end{pgfonlayer}
\end{tikzpicture}
\iff
\begin{tikzpicture}
	\begin{pgfonlayer}{nodelayer}
		\node [style=map] (0) at (2, -2) {$V$};
		\node [style=none] (1) at (1.75, -1.25) {};
		\node [style=none] (2) at (2.25, -1.25) {};
		\node [style=none] (3) at (1.75, -2.75) {};
		\node [style=none] (4) at (2.25, -2.75) {};
	\end{pgfonlayer}
	\begin{pgfonlayer}{edgelayer}
		\draw [in=120, out=-90] (1.center) to (0);
		\draw [in=-90, out=60] (0) to (2.center);
		\draw [in=-60, out=90] (4.center) to (0);
		\draw [in=90, out=-120] (0) to (3.center);
	\end{pgfonlayer}
\end{tikzpicture}
=
\begin{tikzpicture}
	\begin{pgfonlayer}{nodelayer}
		\node [style=map] (0) at (2.5, -1.75) {$V$};
		\node [style=none] (1) at (2, -0.5) {};
		\node [style=none] (2) at (3, -0.5) {};
		\node [style=none] (3) at (3, -2.5) {};
		\node [style=X] (4) at (1.75, -3) {};
		\node [style=Z] (5) at (2.5, -3) {};
		\node [style=none] (6) at (2.5, -0.75) {};
		\node [style=none] (7) at (1.5, -0.75) {};
		\node [style=none] (8) at (1.75, -2) {};
		\node [style=X] (9) at (1.5, -0.75) {};
		\node [style=Z] (10) at (2.5, -0.75) {};
		\node [style=none] (11) at (0.5, -3.25) {};
		\node [style=none] (12) at (1, -3.25) {};
	\end{pgfonlayer}
	\begin{pgfonlayer}{edgelayer}
		\draw [in=120, out=-90] (1.center) to (0);
		\draw [in=-90, out=60] (0) to (2.center);
		\draw [in=-45, out=90] (3.center) to (0);
		\draw [in=-45, out=150] (4) to (7.center);
		\draw [in=30, out=-90] (3.center) to (5);
		\draw [in=90, out=-45, looseness=1.25] (6.center) to (8.center);
		\draw [in=150, out=-90] (8.center) to (5);
		\draw [in=45, out=-135] (0) to (4);
		\draw [in=90, out=-150] (10) to (12.center);
		\draw [in=90, out=-120] (9) to (11.center);
	\end{pgfonlayer}
\end{tikzpicture}
=
\begin{tikzpicture}
	\begin{pgfonlayer}{nodelayer}
		\node [style=map] (59) at (17, -1.75) {$V^\perp$};
		\node [style=none] (60) at (16.25, -0.75) {};
		\node [style=none] (61) at (17.5, -0.75) {};
		\node [style=none] (62) at (17.5, -2.5) {};
		\node [style=Z] (63) at (16.25, -3) {};
		\node [style=X] (64) at (17, -3) {};
		\node [style=none] (65) at (17, -0.75) {};
		\node [style=none] (66) at (15.75, -0.75) {};
		\node [style=none] (67) at (16.25, -2) {};
		\node [style=none] (68) at (16.5, 0.75) {};
		\node [style=none] (69) at (17.5, 1) {};
		\node [style=none] (70) at (15.5, 0.75) {};
		\node [style=none] (71) at (16, 1) {};
		\node [style=s] (72) at (17.5, -0.75) {};
		\node [style=s] (73) at (17, -0.75) {};
		\node [style=X] (74) at (15.5, 0.75) {};
		\node [style=Z] (75) at (16.5, 0.75) {};
		\node [style=none] (76) at (15, -3.25) {};
		\node [style=none] (77) at (15.5, -3.25) {};
		\node [style=none] (78) at (16.5, 1.5) {};
	\end{pgfonlayer}
	\begin{pgfonlayer}{edgelayer}
		\draw [in=120, out=-90] (60.center) to (59);
		\draw [in=-90, out=60] (59) to (61.center);
		\draw [in=-45, out=90] (62.center) to (59);
		\draw [in=-90, out=135, looseness=0.75] (63) to (66.center);
		\draw [in=30, out=-90] (62.center) to (64);
		\draw [in=90, out=-90] (65.center) to (67.center);
		\draw [in=150, out=-90] (67.center) to (64);
		\draw [in=45, out=-135] (59) to (63);
		\draw [in=-90, out=90, looseness=1.25] (61.center) to (71.center);
		\draw [in=-60, out=90] (65.center) to (70.center);
		\draw [in=270, out=90] (60.center) to (69.center);
		\draw [in=-45, out=90, looseness=1.25] (66.center) to (68.center);
		\draw [in=-135, out=90, looseness=0.50] (76.center) to (74);
		\draw [in=-150, out=90] (77.center) to (75);
	\end{pgfonlayer}
\end{tikzpicture}
=
\begin{tikzpicture}
	\begin{pgfonlayer}{nodelayer}
		\node [style=map] (0) at (4.5, -1.75) {$V^\perp$};
		\node [style=none] (1) at (4, -1) {};
		\node [style=none] (2) at (5, -1) {};
		\node [style=none] (3) at (5, -2.5) {};
		\node [style=X] (4) at (4.75, -3) {};
		\node [style=none] (5) at (3.75, -2.25) {};
		\node [style=none] (6) at (5, 0) {};
		\node [style=none] (7) at (4, 0) {};
		\node [style=s] (8) at (5, -1) {};
		\node [style=none] (9) at (3.25, -3.25) {};
		\node [style=none] (10) at (4, -3.25) {};
		\node [style=Z] (11) at (3.75, -2.25) {};
	\end{pgfonlayer}
	\begin{pgfonlayer}{edgelayer}
		\draw [in=135, out=-90] (1.center) to (0);
		\draw [in=-90, out=60] (0) to (2.center);
		\draw [in=-45, out=90] (3.center) to (0);
		\draw [in=30, out=-90] (3.center) to (4);
		\draw [in=165, out=-15, looseness=1.25] (5.center) to (4);
		\draw [in=-90, out=90, looseness=1.25] (2.center) to (7.center);
		\draw [in=270, out=90] (1.center) to (6.center);
		\draw [in=240, out=90] (10.center) to (0);
		\draw [in=90, out=-150] (11) to (9.center);
	\end{pgfonlayer}
\end{tikzpicture}
=
\begin{tikzpicture}
	\begin{pgfonlayer}{nodelayer}
		\node [style=map] (0) at (2, -2) {$V^\perp$};
		\node [style=none] (1) at (1.75, -1.25) {};
		\node [style=none] (2) at (2.25, -1.25) {};
		\node [style=none] (3) at (1.75, -2.75) {};
		\node [style=none] (4) at (2.25, -2.75) {};
		\node [style=none] (5) at (2.25, -0.5) {};
		\node [style=none] (6) at (1.75, -0.5) {};
		\node [style=none] (7) at (2.25, -3.5) {};
		\node [style=none] (8) at (1.75, -3.5) {};
		\node [style=s] (9) at (2.25, -1.25) {};
		\node [style=s] (10) at (2.25, -2.75) {};
	\end{pgfonlayer}
	\begin{pgfonlayer}{edgelayer}
		\draw [in=120, out=-90] (1.center) to (0);
		\draw [in=-90, out=60] (0) to (2.center);
		\draw [in=-60, out=90] (4.center) to (0);
		\draw [in=90, out=-120] (0) to (3.center);
		\draw [in=90, out=-90] (6.center) to (2.center);
		\draw [in=270, out=90] (1.center) to (5.center);
		\draw [in=270, out=90] (7.center) to (3.center);
		\draw [in=270, out=90] (8.center) to (4.center);
	\end{pgfonlayer}
\end{tikzpicture}
$$
For this reason, we will depict Lagrangian relations as processes, where the input wires are on the bottom and output wires on on the top.

\begin{lemma}
There is a faithful, strong symmetric monoidal functor $L:\LinRel_k\to\Lag\Rel_k$ given by the following action on the generators of $\ih_k$ (i.e.\ doubling, and then changing the colours of one of the copies):
$$
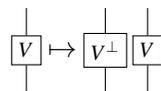
\begin{tikzpicture}
	\begin{pgfonlayer}{nodelayer}
		\node [style=map] (0) at (-3, -1) {$V$};
		\node [style=none] (1) at (-3, -0.25) {};
		\node [style=none] (2) at (-3, -1.75) {};
	\end{pgfonlayer}
	\begin{pgfonlayer}{edgelayer}
		\draw (1.center) to (0);
		\draw (0) to (2.center);
	\end{pgfonlayer}
\end{tikzpicture}
\mapsto
\begin{tikzpicture}
	\begin{pgfonlayer}{nodelayer}
		\node [style=map] (0) at (-3, -1) {$V^\perp$};
		\node [style=none] (1) at (-3, -0.25) {};
		\node [style=none] (2) at (-3, -1.75) {};
		\node [style=map] (3) at (-2.25, -1) {$V$};
		\node [style=none] (4) at (-2.25, -0.25) {};
		\node [style=none] (5) at (-2.25, -1.75) {};
	\end{pgfonlayer}
	\begin{pgfonlayer}{edgelayer}
		\draw (1.center) to (0);
		\draw (0) to (2.center);
		\draw (4.center) to (3);
		\draw (3) to (5.center);
	\end{pgfonlayer}
\end{tikzpicture}
$$
\end{lemma}

To check this is a functor, all we have to show is that it produces Lagrangian relations. This follows immediately from the naturality of $-1$.
This functor is symmetric monoidal and faithful but not full, as for example, the following Lagrangian relation is not in the image of $L$:
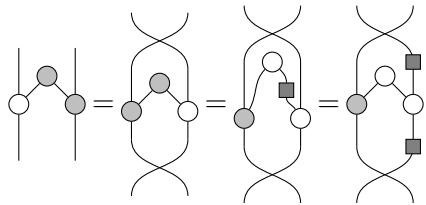
$$
\begin{tikzpicture}
	\begin{pgfonlayer}{nodelayer}
		\node [style=Z] (0) at (0.5, 0) {};
		\node [style=none] (1) at (0.5, 1) {};
		\node [style=none] (2) at (0.5, -1) {};
		\node [style=X] (3) at (1.5, 0) {};
		\node [style=X] (4) at (1, 0.5) {};
		\node [style=none] (5) at (1.5, 1) {};
		\node [style=none] (6) at (1.5, -1) {};
	\end{pgfonlayer}
	\begin{pgfonlayer}{edgelayer}
		\draw (2.center) to (0);
		\draw (0) to (1.center);
		\draw (6.center) to (3);
		\draw (3) to (5.center);
		\draw (3) to (4);
		\draw (4) to (0);
	\end{pgfonlayer}
\end{tikzpicture}
=
\begin{tikzpicture}
	\begin{pgfonlayer}{nodelayer}
		\node [style=Z] (14) at (4.5, -0.25) {};
		\node [style=none] (15) at (4.5, 0.5) {};
		\node [style=none] (16) at (4.5, -0.75) {};
		\node [style=X] (17) at (3.5, -0.25) {};
		\node [style=none] (18) at (3.5, 0.5) {};
		\node [style=none] (19) at (3.5, -0.75) {};
		\node [style=none] (21) at (3.5, 1.5) {};
		\node [style=none] (22) at (4.5, 1.5) {};
		\node [style=none] (23) at (3.5, -1.75) {};
		\node [style=none] (24) at (4.5, -1.75) {};
		\node [style=X] (25) at (4, 0.25) {};
	\end{pgfonlayer}
	\begin{pgfonlayer}{edgelayer}
		\draw (16.center) to (14);
		\draw (14) to (15.center);
		\draw (19.center) to (17);
		\draw (17) to (18.center);
		\draw [in=270, out=90] (18.center) to (22.center);
		\draw [in=270, out=90] (15.center) to (21.center);
		\draw [in=270, out=90] (23.center) to (16.center);
		\draw [in=270, out=90] (24.center) to (19.center);
		\draw (17) to (25);
		\draw (14) to (25);
	\end{pgfonlayer}
\end{tikzpicture}
=
\begin{tikzpicture}
	\begin{pgfonlayer}{nodelayer}
		\node [style=Z] (96) at (21.5, -0.25) {};
		\node [style=none] (97) at (21.5, 0.75) {};
		\node [style=none] (98) at (21.5, -0.75) {};
		\node [style=X] (99) at (20.5, -0.25) {};
		\node [style=none] (100) at (20.5, 0.75) {};
		\node [style=none] (101) at (20.5, -0.75) {};
		\node [style=Z] (102) at (21, 0.75) {};
		\node [style=none] (103) at (20.5, 1.75) {};
		\node [style=none] (104) at (21.5, 1.75) {};
		\node [style=none] (105) at (20.5, -1.75) {};
		\node [style=none] (106) at (21.5, -1.75) {};
		\node [style=s] (107) at (21.25, 0.25) {};
	\end{pgfonlayer}
	\begin{pgfonlayer}{edgelayer}
		\draw (98.center) to (96);
		\draw (96) to (97.center);
		\draw (101.center) to (99);
		\draw (99) to (100.center);
		\draw [in=-135, out=60] (99) to (102);
		\draw [in=270, out=90] (100.center) to (104.center);
		\draw [in=270, out=90] (97.center) to (103.center);
		\draw [in=270, out=90] (105.center) to (98.center);
		\draw [in=270, out=90] (106.center) to (101.center);
		\draw [in=-90, out=120] (96) to (107);
		\draw [in=-45, out=90] (107) to (102);
	\end{pgfonlayer}
\end{tikzpicture}
=
\begin{tikzpicture}
	\begin{pgfonlayer}{nodelayer}
		\node [style=Z] (0) at (2.5, 0) {};
		\node [style=none] (1) at (2.5, 0.75) {};
		\node [style=none] (2) at (2.5, -0.75) {};
		\node [style=X] (3) at (1.5, 0) {};
		\node [style=none] (5) at (1.5, 0.75) {};
		\node [style=none] (6) at (1.5, -0.75) {};
		\node [style=Z] (7) at (2, 0.5) {};
		\node [style=none] (8) at (1.5, 1.75) {};
		\node [style=none] (9) at (2.5, 1.75) {};
		\node [style=none] (10) at (1.5, -1.75) {};
		\node [style=none] (11) at (2.5, -1.75) {};
		\node [style=s] (12) at (2.5, -0.75) {};
		\node [style=s] (13) at (2.5, 0.75) {};
	\end{pgfonlayer}
	\begin{pgfonlayer}{edgelayer}
		\draw (2.center) to (0);
		\draw (0) to (1.center);
		\draw (6.center) to (3);
		\draw (3) to (5.center);
		\draw (3) to (7);
		\draw (0) to (7);
		\draw [in=270, out=90] (5.center) to (9.center);
		\draw [in=270, out=90] (1.center) to (8.center);
		\draw [in=270, out=90] (10.center) to (2.center);
		\draw [in=270, out=90] (11.center) to (6.center);
	\end{pgfonlayer}
\end{tikzpicture}
$$

\section{Generators for Lagrangian relations}
\label{sec:univ}

In this section, we shall give a universal set of generators for $\Lag\Rel_k$, although, we do not directly give a complete set of identities.  Instead we defer to the completeness of the underlying category $\ih_k\cong\LinRel_k$.

Consider the following symplectomorphisms, ie.\ the discrete Fourier transform $F$,  the $a$-shift gate $S_a$ and the controlled-$a$ gate $C_a$:
$$
\left\llbracket
\begin{tikzpicture}
	\begin{pgfonlayer}{nodelayer}
		\node [style=none] (0) at (0.5, 1) {};
		\node [style=none] (1) at (0.5, -0.25) {};
		\node [style=none] (2) at (1, -0.25) {};
		\node [style=none] (3) at (1, 1) {};
		\node [style=s] (4) at (1, 0.5) {};
		\node [style=none] (5) at (0.5, 0.5) {};
	\end{pgfonlayer}
	\begin{pgfonlayer}{edgelayer}
		\draw (4) to (3.center);
		\draw [in=90, out=-90] (4) to (1.center);
		\draw [in=-90, out=90] (2.center) to (5.center);
		\draw (5.center) to (0.center);
	\end{pgfonlayer}
\end{tikzpicture}
\right\rrbracket
=
\begin{bmatrix}
0   & 1 \\
-1  & 0
\end{bmatrix}
\hspace*{.2cm}
\left\llbracket
\begin{tikzpicture}
	\begin{pgfonlayer}{nodelayer}
		\node [style=X] (0) at (0.5, 1.25) {};
		\node [style=Z] (1) at (1.5, -0.25) {};
		\node [style=scalar] (2) at (1, 0.5) {$a$};
		\node [style=none] (3) at (0.5, 1.75) {};
		\node [style=none] (4) at (1.5, 1.75) {};
		\node [style=none] (5) at (1.5, -0.75) {};
		\node [style=none] (6) at (0.5, -0.75) {};
	\end{pgfonlayer}
	\begin{pgfonlayer}{edgelayer}
		\draw (5.center) to (1);
		\draw (1) to (4.center);
		\draw [in=-90, out=135] (1) to (2);
		\draw [in=-45, out=90] (2) to (0);
		\draw (3.center) to (0);
		\draw (0) to (6.center);
	\end{pgfonlayer}
\end{tikzpicture}
\right\rrbracket
=
\begin{bmatrix}
1 &a\\
0 & 1
\end{bmatrix}
\hspace*{.2cm}
\left\llbracket
\begin{tikzpicture}
	\begin{pgfonlayer}{nodelayer}
		\node [style=Z] (430) at (219.75, 0) {};
		\node [style=X] (431) at (220.75, 1.5) {};
		\node [style=none] (432) at (219.75, 0) {};
		\node [style=none] (433) at (221.25, -0.75) {};
		\node [style=none] (434) at (219.25, 2.25) {};
		\node [style=none] (435) at (220.75, 2.25) {};
		\node [style=scalar] (436) at (220.25, 0.75) {$a$};
		\node [style=X] (437) at (217.5, 0) {};
		\node [style=Z] (438) at (218.5, 1.5) {};
		\node [style=none] (439) at (217.5, -0.75) {};
		\node [style=none] (440) at (219, -0.75) {};
		\node [style=none] (441) at (217.25, 2.25) {};
		\node [style=none] (442) at (218.5, 1.5) {};
		\node [style=scalarop] (443) at (218, 0.75) {$a$};
		\node [style=none] (445) at (219.75, -0.75) {};
		\node [style=none] (446) at (218.5, 2.25) {};
	\end{pgfonlayer}
	\begin{pgfonlayer}{edgelayer}
		\draw [in=-105, out=30] (430) to (436);
		\draw [in=-150, out=90] (436) to (431);
		\draw [in=90, out=-60] (431) to (433.center);
		\draw (431) to (435.center);
		\draw [in=135, out=-90, looseness=0.75] (434.center) to (430);
		\draw (439.center) to (437);
		\draw [in=-105, out=30] (437) to (443);
		\draw [in=-150, out=90] (443) to (438);
		\draw [in=90, out=-45, looseness=0.75] (438) to (440.center);
		\draw [in=120, out=-90] (441.center) to (437);
		\draw [in=270, out=90] (442.center) to (446.center);
		\draw [in=270, out=90] (445.center) to (432.center);
	\end{pgfonlayer}
\end{tikzpicture}
\right\rrbracket
=
\begin{bmatrix}
1 & -a & 0 & 0 \\
0 & 1 & 0 & 0 \\
0 & 0 & 1 & 0 \\
0 & 0 & a & 1
\end{bmatrix}
$$

We use the notation $G^{(j)}$ to denote a gate $G$ being applied to wire $j$, and the notation $C_a^{(i,j)}$ to denote the controlled-$a$ gate controlling on wire $i$ targetting wire $j$.

Note the right action of these gates in terms of matrix multiplication of Lagrangian subspaces for any nonzero $a \in k$ (as observed in \cite[p. 4]{aaronson}):

\begin{itemize}
\item
$F^{(i)}$ sets columns $x_i$ to $-z_i$ and $z_i$ to $x_i$.

\item
$S_a^{(i)}$ sets $z_i$ to $z_i+a\cdot x_i$.

\item
$C_a^{(i,j)}$ sets $x_j$ to $x_j- a \cdot x_i$ and $z_i$ to $z_i+a\cdot z_j$.

\end{itemize}

Using these symplectomorphisms regarded as Lagrangian relations, we have:

\begin{theorem}
\label{theorem:generators}
For any field $k$ the maps in $L(\LinRel_k)$ as well as $F$ and $S_a$ for all $a \in k$ generate $\Lag\Rel_k$.
\end{theorem}
This is proved by using these symplectomorphisms to reduce a Lagrangian relation by Guassian elimination to the state
$L(
\begin{tikzpicture}[scale=.5]
	\begin{pgfonlayer}{nodelayer}
		\node [style=X] (0) at (6, 3) {};
		\node [style=none] (1) at (6, 3.5) {};
	\end{pgfonlayer}
	\begin{pgfonlayer}{edgelayer}
		\draw (0) to (1.center);
	\end{pgfonlayer}
\end{tikzpicture}
^{\otimes n}$).

We can also give a presentation of this category which is very close to Selinger's CPM construction~\cite{cpm}. There are several equivalent ways to define the CPM construction. For our purposes, the most convenient one is the presentation used in both ~\cite{pqp,cqm1}, which defines $\CPM[\X]$ as the subcategory of a dagger compact closed category $\X$ whose objects are of the form $A^* \otimes A$ for $A \in \X$ and whose morphisms are generated by (i) `pure' morphisms, i.e.\ morphisms of the form $f_* \otimes f$ for $f \in \X$ and a covariant functor $(\_)_*$, and (ii) a `discard' morphism $d_A$ for every $A \in \X$ given by the counit $d_A := \epsilon_A : A^* \otimes A \to I$ of the compact closed structure on $A$.

We nearly obtain such a presentation for $\Lag\Rel_k$ using the covariant functor $(-)^\perp$ to define pure morphisms, with the only caveat being that we need to consider a family of discard morphisms, with each discard morphism being parametrised by a field element.

\begin{theorem}
\label{theorem:unbiased}
$\Lag\Rel_k$ is the monoidal subcategory of $\LinRel_k$ whose objects are of the form $k^n \oplus k^n$, for all natural numbers $n$, and whose morphisms are generated by \textit{pure} morphisms of the form $V^\perp \oplus V$ for $V \in \LinRel_k$ and for each $a \in k$, a `discard' morphism:
$$
d_a := \begin{tikzpicture}
	\begin{pgfonlayer}{nodelayer}
		\node [style=X] (0) at (0, 0.75) {};
		\node [style=scalar] (1) at (0.5, 0) {$a$};
		\node [style=none] (2) at (0.5, -0.75) {};
		\node [style=none] (3) at (-0.5, 0) {};
		\node [style=none] (4) at (-0.5, -0.75) {};
	\end{pgfonlayer}
	\begin{pgfonlayer}{edgelayer}
		\draw [in=-30, out=90] (1) to (0);
		\draw [in=90, out=-150] (0) to (3.center);
		\draw (4.center) to (3.center);
		\draw (2.center) to (1);
	\end{pgfonlayer}
\end{tikzpicture}

$$
\end{theorem}

\begin{proof}
  We just have to show that $F$ and $S_a$ can be constructed using these generators. The $S_a$ gate and its colour-reversed version $V_a$ can be obtained by composing a pure morphism with $d_a$ and $d_{-a}$, respectively:
$$
\begin{tikzpicture}
	\begin{pgfonlayer}{nodelayer}
		\node [style=none] (0) at (1, -1.25) {};
		\node [style=Z] (1) at (1, -0.75) {};
		\node [style=X] (2) at (-0.75, -0.75) {};
		\node [style=X] (3) at (-0.375, 0.75) {};
		\node [style=none] (4) at (-0.75, -1.25) {};
		\node [style=none] (5) at (-0.75, 1.25) {};
		\node [style=none] (6) at (1, 1.25) {};
		\node [style=scalar] (7) at (0.5, 0) {$a$};
		\node [style=none] (8) at (-1.25, 0) {};
		\node [style=none] (9) at (1.75, 0) {$=$};
		\node [style=none] (10) at (3.5, -1) {};
		\node [style=Z] (11) at (3.5, -0.5) {};
		\node [style=X] (12) at (2.5, 0.5) {};
		\node [style=none] (13) at (2.5, -1) {};
		\node [style=scalar] (14) at (3, 0) {$a$};
		\node [style=none] (15) at (2.5, 1) {};
		\node [style=none] (16) at (3.5, 1) {};
	\end{pgfonlayer}
	\begin{pgfonlayer}{edgelayer}
		\draw (1) to (0.center);
		\draw (1) to (6.center);
		\draw (2) to (5.center);
		\draw (2) to (4.center);
		\draw [in=-90, out=165] (1) to (7);
		\draw [in=0, out=90] (7) to (3);
		\draw [in=-90, out=150] (2) to (8.center);
		\draw [in=-180, out=90] (8.center) to (3);
		\draw (11) to (10.center);
		\draw [in=-90, out=165] (11) to (14);
		\draw [in=-15, out=90] (14) to (12);
		\draw (13.center) to (12);
		\draw (11) to (16.center);
		\draw (12) to (15.center);
	\end{pgfonlayer}
\end{tikzpicture}
\ = S_a
\qquad\qquad
\begin{tikzpicture}
	\begin{pgfonlayer}{nodelayer}
		\node [style=none] (0) at (1, -1.25) {};
		\node [style=X] (1) at (1, -0.75) {};
		\node [style=Z] (2) at (-0.75, -0.75) {};
		\node [style=Z] (3) at (-0.375, 0.75) {};
		\node [style=none] (4) at (-0.75, -1.25) {};
		\node [style=none] (5) at (-0.75, 1.25) {};
		\node [style=none] (6) at (1, 1.25) {};
		\node [style=scalar] (7) at (0.5, 0) {$a$};
		\node [style=none] (8) at (-1.25, 0) {};
		\node [style=none] (9) at (1.75, 0) {$=$};
		\node [style=none] (10) at (3.5, -1) {};
		\node [style=X] (11) at (3.5, -0.5) {};
		\node [style=Z] (12) at (2.5, 0.5) {};
		\node [style=none] (13) at (2.5, -1) {};
		\node [style=scalar] (14) at (3, 0) {$a$};
		\node [style=none] (15) at (2.5, 1) {};
		\node [style=none] (16) at (3.5, 1) {};
		\node [style=none] (17) at (-2.75, -1.25) {};
		\node [style=X] (18) at (-2.75, -0.75) {};
		\node [style=Z] (19) at (-4.5, -0.75) {};
		\node [style=X] (20) at (-4.125, 0.75) {};
		\node [style=none] (21) at (-4.5, -1.25) {};
		\node [style=none] (22) at (-4.5, 1.25) {};
		\node [style=none] (23) at (-2.75, 1.25) {};
		\node [style=scalar] (24) at (-3.25, 0) {-$a$};
		\node [style=none] (25) at (-5, 0) {};
		\node [style=none] (26) at (-2, 0) {$=$};
	\end{pgfonlayer}
	\begin{pgfonlayer}{edgelayer}
		\draw (1) to (0.center);
		\draw (1) to (6.center);
		\draw (2) to (5.center);
		\draw (2) to (4.center);
		\draw [in=-90, out=165] (1) to (7);
		\draw [in=0, out=90] (7) to (3);
		\draw [in=-90, out=150] (2) to (8.center);
		\draw [in=-180, out=90] (8.center) to (3);
		\draw (11) to (10.center);
		\draw [in=-90, out=165] (11) to (14);
		\draw [in=-15, out=90] (14) to (12);
		\draw (13.center) to (12);
		\draw (11) to (16.center);
		\draw (12) to (15.center);
		\draw (18) to (17.center);
		\draw (18) to (23.center);
		\draw (19) to (22.center);
		\draw (19) to (21.center);
		\draw [in=-90, out=165] (18) to (24);
		\draw [in=0, out=90] (24) to (20);
		\draw [in=-90, out=150] (19) to (25.center);
		\draw [in=-180, out=90] (25.center) to (20);
	\end{pgfonlayer}
\end{tikzpicture}
\  =: V_a
$$

We can then obtain $F$ as $S_1 \circ V_1 \circ S_1$, which can be proven as a variation of the familiar `3 CNOT' rule for quantum circuits (see e.g.\ ~\cite[\S 3.2.1]{coecke2008interacting}):
$$
S_1 \circ V_1 \circ S_1
=
\begin{tikzpicture}[xscale=-1]
	\begin{pgfonlayer}{nodelayer}
		\node [style=Z] (0) at (3, 1.25) {};
		\node [style=X] (1) at (4, 1.75) {};
		\node [style=Z] (2) at (4, 2.5) {};
		\node [style=X] (3) at (3, 2) {};
		\node [style=Z] (4) at (3, 2.75) {};
		\node [style=X] (5) at (4, 3.25) {};
		\node [style=none] (6) at (3, 4.5) {};
		\node [style=none] (7) at (4, 4.5) {};
		\node [style=none] (8) at (3, 0) {};
		\node [style=none] (9) at (4, 0) {};
	\end{pgfonlayer}
	\begin{pgfonlayer}{edgelayer}
		\draw (6.center) to (4);
		\draw (4) to (3);
		\draw (0) to (3);
		\draw (8.center) to (0);
		\draw (9.center) to (1);
		\draw (1) to (2);
		\draw (2) to (5);
		\draw (5) to (7.center);
		\draw (3) to (2);
		\draw (4) to (5);
		\draw (0) to (1);
	\end{pgfonlayer}
\end{tikzpicture}
=
\begin{tikzpicture}[xscale=-1]
	\begin{pgfonlayer}{nodelayer}
		\node [style=Z] (0) at (3, 1.25) {};
		\node [style=X] (1) at (4, 0.75) {};
		\node [style=Z] (2) at (4, 2.5) {};
		\node [style=X] (3) at (3, 2) {};
		\node [style=Z] (4) at (3, 3.75) {};
		\node [style=X] (5) at (4, 3.25) {};
		\node [style=none] (6) at (3, 4.5) {};
		\node [style=none] (7) at (4, 4.5) {};
		\node [style=none] (8) at (3, 0) {};
		\node [style=none] (9) at (4, 0) {};
		\node [style=s] (10) at (3.5, 3.5) {};
		\node [style=s] (11) at (3.5, 1) {};
	\end{pgfonlayer}
	\begin{pgfonlayer}{edgelayer}
		\draw (6.center) to (4);
		\draw (4) to (3);
		\draw (0) to (3);
		\draw (8.center) to (0);
		\draw (9.center) to (1);
		\draw (1) to (2);
		\draw (2) to (5);
		\draw (5) to (7.center);
		\draw (3) to (2);
		\draw (1) to (11);
		\draw (11) to (0);
		\draw (5) to (10);
		\draw (10) to (4);
	\end{pgfonlayer}
\end{tikzpicture}
=
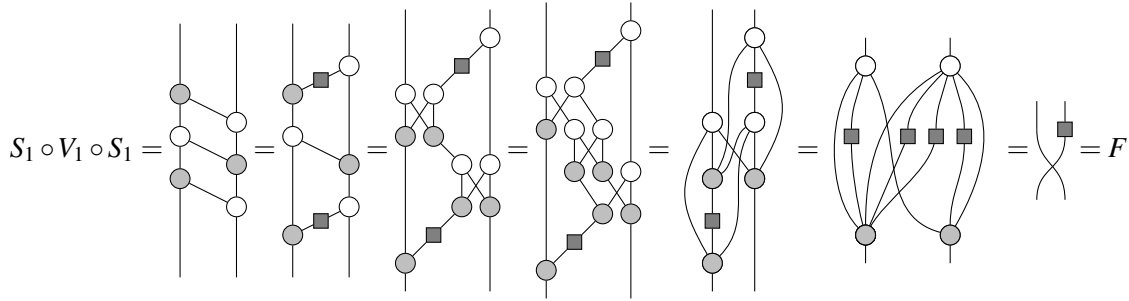
\begin{tikzpicture}[xscale=-1]
	\begin{pgfonlayer}{nodelayer}
		\node [style=X] (28) at (6.5, 0) {};
		\node [style=Z] (31) at (5, 4) {};
		\node [style=none] (33) at (5, 4.5) {};
		\node [style=none] (34) at (6.5, 4.5) {};
		\node [style=none] (35) at (5, -0.5) {};
		\node [style=none] (36) at (6.5, -0.5) {};
		\node [style=s] (37) at (6, 0.5) {};
		\node [style=s] (38) at (5.5, 3.5) {};
		\node [style=X] (39) at (6, 2.25) {};
		\node [style=Z] (40) at (6, 3) {};
		\node [style=X] (41) at (6.5, 2.25) {};
		\node [style=Z] (42) at (6.5, 3) {};
		\node [style=X] (43) at (5, 1) {};
		\node [style=Z] (44) at (5, 1.75) {};
		\node [style=X] (45) at (5.5, 1) {};
		\node [style=Z] (46) at (5.5, 1.75) {};
	\end{pgfonlayer}
	\begin{pgfonlayer}{edgelayer}
		\draw (33.center) to (31);
		\draw (36.center) to (28);
		\draw (28) to (37);
		\draw (31) to (38);
		\draw (40) to (41);
		\draw (41) to (42);
		\draw (40) to (39);
		\draw (39) to (42);
		\draw (44) to (45);
		\draw (45) to (46);
		\draw (44) to (43);
		\draw (43) to (46);
		\draw (34.center) to (42);
		\draw (40) to (38);
		\draw (46) to (39);
		\draw (41) to (28);
		\draw (37) to (45);
		\draw (35.center) to (43);
		\draw (31) to (44);
	\end{pgfonlayer}
\end{tikzpicture}
=
\begin{tikzpicture}[xscale=-1]
	\begin{pgfonlayer}{nodelayer}
		\node [style=X] (67) at (11.5, 0) {};
		\node [style=Z] (68) at (10, 4.25) {};
		\node [style=none] (69) at (10, 4.75) {};
		\node [style=none] (70) at (11.5, 4.75) {};
		\node [style=none] (71) at (10, -0.5) {};
		\node [style=none] (72) at (11.5, -0.5) {};
		\node [style=s] (73) at (11, 0.5) {};
		\node [style=s] (74) at (10.5, 3.75) {};
		\node [style=Z] (75) at (11, 3.25) {};
		\node [style=X] (76) at (11.5, 2.5) {};
		\node [style=Z] (77) at (11.5, 3.25) {};
		\node [style=X] (78) at (10, 1) {};
		\node [style=Z] (79) at (10, 1.75) {};
		\node [style=X] (80) at (10.5, 1) {};
		\node [style=X] (81) at (10.5, 1.75) {};
		\node [style=Z] (82) at (10.5, 2.5) {};
		\node [style=X] (83) at (11, 1.75) {};
		\node [style=Z] (84) at (11, 2.5) {};
	\end{pgfonlayer}
	\begin{pgfonlayer}{edgelayer}
		\draw (69.center) to (68);
		\draw (72.center) to (67);
		\draw (67) to (73);
		\draw (68) to (74);
		\draw (75) to (76);
		\draw (76) to (77);
		\draw (79) to (80);
		\draw (79) to (78);
		\draw (70.center) to (77);
		\draw (75) to (74);
		\draw (76) to (67);
		\draw (73) to (80);
		\draw (71.center) to (78);
		\draw (68) to (79);
		\draw (82) to (83);
		\draw (83) to (84);
		\draw (82) to (81);
		\draw (81) to (84);
		\draw (78) to (81);
		\draw (80) to (83);
		\draw (84) to (77);
		\draw (75) to (82);
	\end{pgfonlayer}
\end{tikzpicture}
=
\begin{tikzpicture}[xscale=-1]
	\begin{pgfonlayer}{nodelayer}
		\node [style=X] (0) at (11, 0.75) {};
		\node [style=Z] (1) at (10.25, 4.75) {};
		\node [style=none] (2) at (10.25, 5.25) {};
		\node [style=none] (3) at (11, 5.25) {};
		\node [style=none] (4) at (10.25, 0.25) {};
		\node [style=none] (5) at (11, 0.25) {};
		\node [style=s] (6) at (11, 1.5) {};
		\node [style=s] (7) at (10.25, 4) {};
		\node [style=Z] (8) at (10.25, 3.25) {};
		\node [style=X] (9) at (11, 0.75) {};
		\node [style=Z] (10) at (11, 3.25) {};
		\node [style=X] (11) at (10.25, 2.25) {};
		\node [style=Z] (12) at (10.25, 4.75) {};
		\node [style=X] (13) at (11, 2.25) {};
		\node [style=X] (14) at (10.25, 2.25) {};
		\node [style=Z] (15) at (10.25, 3.25) {};
		\node [style=X] (16) at (11, 2.25) {};
		\node [style=Z] (17) at (11, 3.25) {};
	\end{pgfonlayer}
	\begin{pgfonlayer}{edgelayer}
		\draw (2.center) to (1);
		\draw (5.center) to (0);
		\draw (0) to (6);
		\draw (1) to (7);
		\draw [in=135, out=-60] (8) to (9);
		\draw [bend right, looseness=1.25] (9) to (10);
		\draw [in=120, out=-45] (12) to (13);
		\draw [bend right, looseness=1.25] (12) to (11);
		\draw (3.center) to (10);
		\draw (8) to (7);
		\draw (6) to (13);
		\draw (4.center) to (11);
		\draw [in=150, out=-30, looseness=0.75] (15) to (16);
		\draw (16) to (17);
		\draw (15) to (14);
		\draw (14) to (17);
	\end{pgfonlayer}
\end{tikzpicture}
=
\begin{tikzpicture}[xscale=-1]
	\begin{pgfonlayer}{nodelayer}
		\node [style=X] (121) at (19.5, 1.25) {};
		\node [style=Z] (122) at (18, 4.25) {};
		\node [style=none] (123) at (18, 4.75) {};
		\node [style=none] (124) at (19.5, 4.75) {};
		\node [style=none] (125) at (18, 0.75) {};
		\node [style=none] (126) at (19.5, 0.75) {};
		\node [style=X] (128) at (19.5, 1.25) {};
		\node [style=Z] (129) at (19.5, 4.25) {};
		\node [style=X] (130) at (18, 1.25) {};
		\node [style=Z] (131) at (18, 4.25) {};
		\node [style=X] (132) at (19.5, 1.25) {};
		\node [style=X] (133) at (18, 1.25) {};
		\node [style=X] (134) at (19.5, 1.25) {};
		\node [style=Z] (135) at (19.5, 4.25) {};
		\node [style=s] (136) at (17.75, 3) {};
		\node [style=s] (138) at (18.75, 3) {};
		\node [style=s] (139) at (19.75, 3) {};
		\node [style=s] (140) at (18.25, 3) {};
	\end{pgfonlayer}
	\begin{pgfonlayer}{edgelayer}
		\draw (123.center) to (122);
		\draw (126.center) to (121);
		\draw [bend right=45, looseness=0.75] (128) to (129);
		\draw [in=120, out=-135] (131) to (130);
		\draw (124.center) to (129);
		\draw (125.center) to (130);
		\draw [in=-120, out=15] (133) to (135);
		\draw [in=90, out=-105] (131) to (136);
		\draw [in=90, out=-45, looseness=0.75] (131) to (138);
		\draw [in=90, out=-90] (136) to (133);
		\draw (128) to (138);
		\draw [in=-30, out=90] (134) to (131);
		\draw [in=-75, out=120, looseness=0.75] (134) to (140);
		\draw [in=285, out=90] (140) to (131);
		\draw [in=-90, out=75, looseness=0.75] (134) to (139);
		\draw [in=-75, out=90] (139) to (135);
	\end{pgfonlayer}
\end{tikzpicture}
=
\begin{tikzpicture}
	\begin{pgfonlayer}{nodelayer}
		\node [style=none] (0) at (24, 4.5) {};
		\node [style=none] (1) at (23.5, 4) {};
		\node [style=none] (2) at (24, 2.75) {};
		\node [style=none] (3) at (23.5, 2.75) {};
		\node [style=s] (4) at (24, 4) {};
		\node [style=none] (5) at (23.5, 4.5) {};
	\end{pgfonlayer}
	\begin{pgfonlayer}{edgelayer}
		\draw [in=-90, out=90, looseness=1.25] (2.center) to (1.center);
		\draw (1.center) to (5.center);
		\draw (0.center) to (4);
		\draw [in=90, out=-90, looseness=1.25] (4) to (3.center);
	\end{pgfonlayer}
\end{tikzpicture}
=
F
$$
\end{proof}

In the ZX-calculus literature, this decomposition of the Fourier transform is known as {\it Euler decomposition} \cite{duncan2009graph}.
A variant of this decomposition is given in \cite[p.\ 6]{control}, although in the context of plain old linear relations instead of Lagrangian relations,  an antipode is missing in their case.  A similar observation was made in \cite[(34)]{ranchin2014depicting} in terms of qudit controlled boost gates; however, the connection to phase-shift gates and Euler decomposition was not made.

From Theorem~\ref{theorem:generators}, we know that we can build any Lagrangian relation using pure Lagrangian relations and discard maps. Since the former is closed under composition and monoidal product, the following can be shown immediately from string diagram deformation.

\begin{corollary}[Purification]\label{cor:pure}
  Any Lagrangian relation can be written in the following form, for $V$ a linear relation:
  \ctikzfig{purification}
\end{corollary}

In the case when we are working with prime fields, then Lagrangian relations are exactly an instance of the CPM construction. Namely, in the category of linear relations, $(-)^*$ is given by relational converse, so we can define a dagger functor $(-)^\dagger := ((-)^\perp)^*$ such that $(-)_* = (-)^\perp$. It only remains to show that all of the discarding maps arise from a single fixed cap. This can be done as follows, for $k = \F_p$:
$$
\begin{tikzpicture}
	\begin{pgfonlayer}{nodelayer}
		\node [style=X] (3) at (3, 2) {};
		\node [style=scalar] (7) at (3.5, 1.5) {$n$};
		\node [style=none] (8) at (2.5, 1.5) {};
		\node [style=none] (9) at (2.5, 1) {};
		\node [style=none] (10) at (3.5, 1) {};
	\end{pgfonlayer}
	\begin{pgfonlayer}{edgelayer}
		\draw [in=-15, out=90] (7) to (3);
		\draw [in=-165, out=90] (8.center) to (3);
		\draw (9.center) to (8.center);
		\draw (10.center) to (7);
	\end{pgfonlayer}
\end{tikzpicture}
=
\begin{tikzpicture}
	\begin{pgfonlayer}{nodelayer}
		\node [style=X] (0) at (5.25, 1.75) {};
		\node [style=Z] (1) at (6.25, 0.5) {};
		\node [style=none] (2) at (4.25, 0) {};
		\node [style=none] (3) at (6.25, 0) {};
		\node [style=none] (4) at (5.75, 1.25) {$\iddots$};
		\node [style=none] (5) at (5.9, 1) {$n$};
		\node [style=X] (6) at (4.75, 2.25) {};
	\end{pgfonlayer}
	\begin{pgfonlayer}{edgelayer}
		\draw [in=0, out=90, looseness=1.25] (1) to (0);
		\draw [in=-90, out=165, looseness=1.25] (1) to (0);
		\draw [in=-120, out=90] (2.center) to (6);
		\draw [in=105, out=-15] (6) to (0);
		\draw (3.center) to (1);
	\end{pgfonlayer}
\end{tikzpicture}
=
\begin{tikzpicture}
	\begin{pgfonlayer}{nodelayer}
		\node [style=none] (21) at (6.25, 2) {};
		\node [style=none] (22) at (5.25, 2.25) {};
		\node [style=X] (23) at (5.25, 2.25) {};
		\node [style=Z] (24) at (6.25, 2) {};
		\node [style=none] (25) at (5, 3) {};
		\node [style=none] (26) at (6.5, 2.75) {};
		\node [style=Z] (28) at (6.5, 2.75) {};
		\node [style=none] (30) at (6.25, 0.75) {};
		\node [style=none] (31) at (5.25, 1) {};
		\node [style=X] (32) at (5.25, 1) {};
		\node [style=Z] (33) at (6.25, 0.75) {};
		\node [style=none] (34) at (5.25, 0.25) {};
		\node [style=none] (35) at (6.25, 0.25) {};
		\node [style=none] (36) at (5.7, 1.6) {$\vdots$};
		\node [style=none] (37) at (5.95, 1.55) {$n$};
		\node [style=X] (38) at (5, 3) {};
	\end{pgfonlayer}
	\begin{pgfonlayer}{edgelayer}
		\draw [in=-90, out=135, looseness=0.75] (23) to (25.center);
		\draw [in=-90, out=60] (24) to (26.center);
		\draw (34.center) to (32);
		\draw (32) to (23);
		\draw (35.center) to (33);
		\draw (33) to (24);
		\draw (24) to (23);
		\draw (33) to (32);
	\end{pgfonlayer}
\end{tikzpicture}
=
\begin{tikzpicture}
	\begin{pgfonlayer}{nodelayer}
		\node [style=X] (622) at (274.75, 2.25) {};
		\node [style=none] (623) at (275.5, 1.25) {};
		\node [style=none] (624) at (274.25, 1.25) {};
		\node [style=X] (625) at (274.25, 1.25) {};
		\node [style=Z] (626) at (275.5, 1.25) {};
		\node [style=none] (627) at (274.25, 2.5) {};
		\node [style=none] (628) at (275.5, 2.5) {};
		\node [style=X] (629) at (274.25, 2.5) {};
		\node [style=Z] (630) at (275.5, 2.5) {};
		\node [style=X] (631) at (274.75, 0.75) {};
		\node [style=none] (632) at (275.5, -0.25) {};
		\node [style=none] (633) at (274.25, -0.25) {};
		\node [style=X] (634) at (274.25, -0.25) {};
		\node [style=Z] (635) at (275.5, -0.25) {};
		\node [style=none] (636) at (274.25, -1) {};
		\node [style=none] (637) at (275.5, -1) {};
		\node [style=none] (638) at (274.8, 1.59) {$\vdots$};
		\node [style=none] (639) at (275.05, 1.5) {$n$};
	\end{pgfonlayer}
	\begin{pgfonlayer}{edgelayer}
		\draw [in=-15, out=120] (623.center) to (622);
		\draw [in=-165, out=165, looseness=1.50] (624.center) to (622);
		\draw (625) to (627.center);
		\draw (626) to (628.center);
		\draw [in=-30, out=120] (632.center) to (631);
		\draw [in=-165, out=150, looseness=1.50] (633.center) to (631);
		\draw (636.center) to (634);
		\draw (634) to (625);
		\draw (637.center) to (635);
		\draw (635) to (626);
	\end{pgfonlayer}
\end{tikzpicture}
$$

\begin{corollary}
\label{cor}
For $p$ prime, $\Lag\Rel_{\F_p} \cong \CPM[\LinRel_{\F_p}]$.
\end{corollary}

\section{Affine Lagrangian relations}
\label{sec:aff}

Affine Lagrangian relations are perhaps of more practical interest than plain old Lagrangian relations.  As we will discuss in this section, these give a semantics for qudit stabilizer circuits as well as certain electrical circuits.  We use our universal set of generators for Lagrangian relations as well as the presentation for affine relations to get a universal set of generators for affine Lagrangian relations.

\begin{definition}[{\cite[\S A]{affine}}]
Let $\aih_k$ denote the the prop presented by $\ih_k$ in addition to the generator
$\begin{tikzpicture}[scale=.7]
	\begin{pgfonlayer}{nodelayer}
		\node [style=X] (0) at (0.25, 0) {$1$};
		\node [style=none] (1) at (0.25, 0.35) {};
	\end{pgfonlayer}
	\begin{pgfonlayer}{edgelayer}
		\draw (0) to (1.center);
	\end{pgfonlayer}
\end{tikzpicture}$
 and three equations:
$$
\begin{tikzpicture}
	\begin{pgfonlayer}{nodelayer}
		\node [style=X] (0) at (0.75, 0) {$1$};
		\node [style=X] (1) at (0.75, 0.5) {};
		\node [style=none] (2) at (1.25, 1) {};
		\node [style=none] (3) at (1.25, -0.5) {};
	\end{pgfonlayer}
	\begin{pgfonlayer}{edgelayer}
		\draw (0) to (1);
		\draw (3.center) to (2.center);
	\end{pgfonlayer}
\end{tikzpicture}
=
\begin{tikzpicture}
	\begin{pgfonlayer}{nodelayer}
		\node [style=X] (0) at (0.75, 0) {$1$};
		\node [style=X] (1) at (0.75, 0.5) {};
		\node [style=none] (2) at (1.25, 1) {};
		\node [style=none] (3) at (1.25, -0.5) {};
		\node [style=Z] (4) at (1.25, 0.5) {};
		\node [style=Z] (5) at (1.25, 0) {};
	\end{pgfonlayer}
	\begin{pgfonlayer}{edgelayer}
		\draw (0) to (1);
		\draw (4) to (2.center);
		\draw (3.center) to (5);
	\end{pgfonlayer}
\end{tikzpicture}
\hspace*{.5cm}
\begin{tikzpicture}
	\begin{pgfonlayer}{nodelayer}
		\node [style=X] (0) at (0.25, 0) {$1$};
		\node [style=Z] (1) at (0.25, 0.5) {};
		\node [style=none] (2) at (0, 1) {};
		\node [style=none] (3) at (0.5, 1) {};
	\end{pgfonlayer}
	\begin{pgfonlayer}{edgelayer}
		\draw (0) to (1);
		\draw [in=-90, out=45] (1) to (3.center);
		\draw [in=135, out=-90] (2.center) to (1);
	\end{pgfonlayer}
\end{tikzpicture}
=
\begin{tikzpicture}
	\begin{pgfonlayer}{nodelayer}
		\node [style=X] (0) at (0, 0.5) {$1$};
		\node [style=none] (2) at (0, 1) {};
		\node [style=none] (3) at (0.5, 1) {};
		\node [style=X] (4) at (0.5, 0.5) {$1$};
	\end{pgfonlayer}
	\begin{pgfonlayer}{edgelayer}
		\draw (0) to (2.center);
		\draw (4) to (3.center);
	\end{pgfonlayer}
\end{tikzpicture}
\hspace*{.5cm}
\begin{tikzpicture}
	\begin{pgfonlayer}{nodelayer}
		\node [style=X] (0) at (0.75, 0) {$1$};
		\node [style=Z] (1) at (0.75, 0.5) {};
	\end{pgfonlayer}
	\begin{pgfonlayer}{edgelayer}
		\draw (0) to (1);
	\end{pgfonlayer}
\end{tikzpicture}
=
$$
\end{definition}

The following was stated slightly differently in the original paper:
\begin{definition}[{\cite[Definition 5]{affine}}]
Let $\Aff\Rel_k$ denote prop whose morphisms $n\to m$ are the (possibly empty) affine subspaces of $k^{n}\oplus k^{m}$, with composition given by relational composition and tensor product given by the direct sum.
\end{definition}

\begin{theorem}[{\cite[Thm.\ 17]{affine}}]
$\aih_k$ is a presentation of $\Aff\Rel_k$.
\end{theorem}
Because the equation on the left below holds,  we can use the phased-spider notation (as in the ZX-calculus), so that for all $a,b \in \F_p$: 
$$
\begin{tikzpicture}
	\begin{pgfonlayer}{nodelayer}
		\node [style=X] (0) at (0, 0) {};
		\node [style=none] (1) at (0, 0.75) {};
		\node [style=scalar] (2) at (-0.5, -0.75) {$a$};
		\node [style=scalar] (3) at (0.5, -0.75) {$b$};
		\node [style=X] (4) at (-0.5, -1.5) {$1$};
		\node [style=X] (5) at (0.5, -1.5) {$1$};
	\end{pgfonlayer}
	\begin{pgfonlayer}{edgelayer}
		\draw (0) to (1.center);
		\draw (4) to (2);
		\draw [in=-150, out=90] (2) to (0);
		\draw [in=90, out=-30] (0) to (3);
		\draw (3) to (5);
	\end{pgfonlayer}
\end{tikzpicture}
=
\begin{tikzpicture}
	\begin{pgfonlayer}{nodelayer}
		\node [style=X] (6) at (2, 0) {};
		\node [style=none] (7) at (2, 0.75) {};
		\node [style=scalar] (8) at (1.5, -0.75) {$a$};
		\node [style=scalar] (9) at (2.5, -0.75) {$b$};
		\node [style=X] (10) at (2, -2) {$1$};
		\node [style=Z] (11) at (2, -1.5) {};
	\end{pgfonlayer}
	\begin{pgfonlayer}{edgelayer}
		\draw (6) to (7.center);
		\draw [in=-150, out=90] (8) to (6);
		\draw [in=90, out=-30] (6) to (9);
		\draw (11) to (10);
		\draw [in=-90, out=30] (11) to (9);
		\draw [in=-90, out=150] (11) to (8);
	\end{pgfonlayer}
\end{tikzpicture}
=
\begin{tikzpicture}
	\begin{pgfonlayer}{nodelayer}
		\node [style=none] (7) at (2, 0.75) {};
		\node [style=scalar] (8) at (2, 0) {$a+b$};
		\node [style=X] (10) at (2, -0.75) {$1$};
	\end{pgfonlayer}
	\begin{pgfonlayer}{edgelayer}
		\draw (10) to (8);
		\draw (8) to (7.center);
	\end{pgfonlayer}
\end{tikzpicture}
\hspace*{.5cm}
\begin{tikzpicture}
	\begin{pgfonlayer}{nodelayer}
		\node [style=X] (0) at (7, 3.5) {$a$};
		\node [style=none] (1) at (6.75, 4) {};
		\node [style=none] (2) at (7.25, 4) {};
		\node [style=none] (3) at (6.75, 3) {};
		\node [style=none] (4) at (7.25, 3) {};
		\node [style=none] (5) at (6.75, 4.25) {};
		\node [style=none] (6) at (7.25, 4.25) {};
		\node [style=none] (7) at (7.25, 2.75) {};
		\node [style=none] (8) at (6.75, 2.75) {};
		\node [style=none] (9) at (7, 4) {$\cdots$};
		\node [style=none] (10) at (7, 3) {$\cdots$};
	\end{pgfonlayer}
	\begin{pgfonlayer}{edgelayer}
		\draw [in=-90, out=135] (0) to (1.center);
		\draw [in=-90, out=45] (0) to (2.center);
		\draw [in=-45, out=90] (4.center) to (0);
		\draw [in=90, out=-135] (0) to (3.center);
		\draw (8.center) to (3.center);
		\draw (7.center) to (4.center);
		\draw (2.center) to (6.center);
		\draw (5.center) to (1.center);
	\end{pgfonlayer}
\end{tikzpicture}
:=
\begin{tikzpicture}
	\begin{pgfonlayer}{nodelayer}
		\node [style=X] (0) at (7, 4) {};
		\node [style=none] (1) at (6.75, 4.5) {};
		\node [style=none] (2) at (7.25, 4.5) {};
		\node [style=none] (3) at (6.75, 2.5) {};
		\node [style=none] (4) at (7.25, 2.5) {};
		\node [style=none] (5) at (6.75, 4.75) {};
		\node [style=none] (6) at (7.25, 4.75) {};
		\node [style=none] (7) at (7.25, 2.25) {};
		\node [style=none] (8) at (6.75, 2.25) {};
		\node [style=none] (9) at (7, 4.5) {$\cdots$};
		\node [style=none] (10) at (7, 2.5) {$\cdots$};
		\node [style=X] (11) at (7.75, 2.75) {$1$};
		\node [style=scalar] (12) at (7.75, 3.25) {$a$};
	\end{pgfonlayer}
	\begin{pgfonlayer}{edgelayer}
		\draw [in=-90, out=135] (0) to (1.center);
		\draw [in=-90, out=45] (0) to (2.center);
		\draw [in=-45, out=90, looseness=0.75] (4.center) to (0);
		\draw [in=90, out=-135, looseness=0.75] (0) to (3.center);
		\draw (8.center) to (3.center);
		\draw (7.center) to (4.center);
		\draw (2.center) to (6.center);
		\draw (5.center) to (1.center);
		\draw (11) to (12);
		\draw [in=-15, out=90, looseness=0.75] (12) to (0);
	\end{pgfonlayer}
\end{tikzpicture}
\hspace{.5cm}
\begin{tikzpicture}
	\begin{pgfonlayer}{nodelayer}
		\node [style=none] (0) at (0.25, 2.5) {};
		\node [style=none] (1) at (1.25, 2.5) {};
		\node [style=none] (2) at (0.25, 0) {};
		\node [style=none] (3) at (1.25, 0) {};
		\node [style=X] (4) at (0, 1.5) {$a$};
		\node [style=X] (5) at (1, 1) {$b$};
		\node [style=none] (6) at (-0.25, 2.5) {};
		\node [style=none] (7) at (0.75, 2.5) {};
		\node [style=none] (8) at (0.75, 0) {};
		\node [style=none] (9) at (-0.25, 0) {};
		\node [style=none] (10) at (0, 2.25) {$\cdots$};
		\node [style=none] (11) at (1, 2.25) {$\cdots$};
		\node [style=none] (12) at (0, 0.25) {$\cdots$};
		\node [style=none] (13) at (1, 0.25) {$\cdots$};
	\end{pgfonlayer}
	\begin{pgfonlayer}{edgelayer}
		\draw [in=-60, out=90, looseness=1.25] (3.center) to (5);
		\draw [in=-90, out=60] (5) to (1.center);
		\draw [in=60, out=-90, looseness=1.25] (0.center) to (4);
		\draw [in=90, out=-60] (4) to (2.center);
		\draw (4) to (5);
		\draw [in=-120, out=90] (9.center) to (4);
		\draw [in=-90, out=120, looseness=1.25] (4) to (6.center);
		\draw [in=-90, out=120] (5) to (7.center);
		\draw [in=-120, out=90, looseness=1.25] (8.center) to (5);
	\end{pgfonlayer}
\end{tikzpicture}
=
\begin{tikzpicture}
	\begin{pgfonlayer}{nodelayer}
		\node [style=none] (14) at (2.75, 2.5) {};
		\node [style=none] (15) at (3.75, 2.5) {};
		\node [style=none] (16) at (2.75, 0) {};
		\node [style=none] (17) at (3.75, 0) {};
		\node [style=X] (18) at (3, 1.25) {};
		\node [style=X] (19) at (3, 1.25) {$\hspace*{.05cm}a+b\hspace*{.05cm}$};
		\node [style=none] (20) at (2.25, 2.5) {};
		\node [style=none] (21) at (3.25, 2.5) {};
		\node [style=none] (22) at (3.25, 0) {};
		\node [style=none] (23) at (2.25, 0) {};
		\node [style=none] (24) at (2.5, 2.25) {$\cdots$};
		\node [style=none] (25) at (3.5, 2.25) {$\cdots$};
		\node [style=none] (26) at (2.5, 0.25) {$\cdots$};
		\node [style=none] (27) at (3.5, 0.25) {$\cdots$};
	\end{pgfonlayer}
	\begin{pgfonlayer}{edgelayer}
		\draw [in=-30, out=90, looseness=1.25] (17.center) to (19);
		\draw [in=-90, out=30, looseness=1.25] (19) to (15.center);
		\draw [in=120, out=-90] (14.center) to (18);
		\draw [in=90, out=-120] (18) to (16.center);
		\draw [in=-150, out=90, looseness=1.25] (23.center) to (18);
		\draw [in=-90, out=150, looseness=1.25] (18) to (20.center);
		\draw [in=-90, out=60] (19) to (21.center);
		\draw [in=-60, out=90] (22.center) to (19);
	\end{pgfonlayer}
\end{tikzpicture}
$$
\begin{definition}
Let $\Aff\Lag\Rel_k$ denote the monoidal category whose objects are symplectic vector spaces with morphisms are generated by the image of $\Lag\Rel_k \xrightarrow{E} \LinRel_k \to \Aff\Rel_k$ as well as all affine shifts and where the  tensor product is the direct sum.
\end{definition}
Because the tensor product is defined in the same way as in $\Lag\Rel_k$, as in Lemma \ref{lemma:strong}, the forgetful functor  $\Aff\Lag\Rel_k\to \Aff\Rel_k$ is faithful, but only {\em strong} monoidal.

\begin{definition}
Let $\alr_k$ denote the monoidal subcategory of $\aih_k$ with objects $2n$, generated by the morphisms in the image of $\Lag\Rel_k\xrightarrow{E} \LinRel_k \cong \ih_k \to \aih_k$ as well as the following generator:
$$
\begin{tikzpicture}
	\begin{pgfonlayer}{nodelayer}
		\node [style=X] (0) at (0, 0) {$1$};
		\node [style=Z] (1) at (0.5, 0) {};
		\node [style=none] (2) at (0, 0.75) {};
		\node [style=none] (3) at (0.5, 0.75) {};
	\end{pgfonlayer}
	\begin{pgfonlayer}{edgelayer}
		\draw (1) to (3.center);
		\draw (2.center) to (0);
	\end{pgfonlayer}
\end{tikzpicture}
$$
\end{definition}

\begin{lemma}
$\alr_k$ is a presentation of $\Aff\Lag\Rel_k$.
\end{lemma}
\begin{proof}
All the affine shifts can be produced from tensoring and composing these two maps on the right:
$$
\begin{tikzpicture}
	\begin{pgfonlayer}{nodelayer}
		\node [style=X] (0) at (0, 0) {$1$};
		\node [style=Z] (1) at (0.5, 0) {};
		\node [style=none] (2) at (0, 0.75) {};
		\node [style=none] (3) at (0.5, 0.75) {};
		\node [style=none] (4) at (0.5, 1.5) {};
		\node [style=none] (5) at (0, 1.5) {};
		\node [style=s] (6) at (0.5, 0.75) {};
	\end{pgfonlayer}
	\begin{pgfonlayer}{edgelayer}
		\draw (1) to (3.center);
		\draw (2.center) to (0);
		\draw [in=270, out=90] (3.center) to (5.center);
		\draw [in=270, out=90] (2.center) to (4.center);
	\end{pgfonlayer}
\end{tikzpicture}
=
\begin{tikzpicture}
	\begin{pgfonlayer}{nodelayer}
		\node [style=X] (0) at (0.5, 0) {$1$};
		\node [style=Z] (1) at (0, 0) {};
		\node [style=none] (2) at (0.5, 0.75) {};
		\node [style=none] (3) at (0, 0.75) {};
	\end{pgfonlayer}
	\begin{pgfonlayer}{edgelayer}
		\draw (1) to (3.center);
		\draw (2.center) to (0);
	\end{pgfonlayer}
\end{tikzpicture} \hspace*{.1cm} \in  \alr_k
\hspace*{.5cm}
\implies
\hspace*{.5cm}
\begin{tikzpicture}
	\begin{pgfonlayer}{nodelayer}
		\node [style=Z] (447) at (224.25, 0.75) {};
		\node [style=X] (448) at (222.75, 0.75) {};
		\node [style=none] (449) at (223.75, -0.25) {};
		\node [style=none] (450) at (222.25, -0.25) {};
		\node [style=none] (451) at (224.25, 1.5) {};
		\node [style=none] (452) at (222.75, 1.5) {};
		\node [style=X] (453) at (223.25, -1) {$1$};
		\node [style=Z] (454) at (224.75, -1) {};
		\node [style=scalar] (455) at (223.25, -0.25) {$a$};
		\node [style=scalarop] (456) at (224.75, -0.25) {$a$};
		\node [style=none] (457) at (223.75, -1.5) {};
		\node [style=none] (458) at (222.25, -1.5) {};
	\end{pgfonlayer}
	\begin{pgfonlayer}{edgelayer}
		\draw [in=90, out=-150] (447) to (449.center);
		\draw [in=-150, out=90] (450.center) to (448);
		\draw (447) to (451.center);
		\draw (448) to (452.center);
		\draw (457.center) to (449.center);
		\draw (458.center) to (450.center);
		\draw (453) to (455);
		\draw (454) to (456);
		\draw [in=-30, out=90] (456) to (447);
		\draw [in=-30, out=90] (455) to (448);
	\end{pgfonlayer}
\end{tikzpicture}
=
\begin{tikzpicture}
	\begin{pgfonlayer}{nodelayer}
		\node [style=X] (17) at (5, 0) {$a$};
		\node [style=none] (18) at (5.5, -1.5) {};
		\node [style=none] (19) at (5, -1.5) {};
		\node [style=none] (20) at (5.5, 1.5) {};
		\node [style=none] (21) at (5, 1.5) {};
	\end{pgfonlayer}
	\begin{pgfonlayer}{edgelayer}
		\draw (19.center) to (17);
		\draw (17) to (21.center);
		\draw (18.center) to (20.center);
	\end{pgfonlayer}
\end{tikzpicture},
\hspace*{.5cm}
\begin{tikzpicture}
	\begin{pgfonlayer}{nodelayer}
		\node [style=Z] (459) at (226.25, 0.75) {};
		\node [style=X] (460) at (227.75, 0.75) {};
		\node [style=none] (461) at (226.75, -0.25) {};
		\node [style=none] (462) at (228.25, -0.25) {};
		\node [style=none] (463) at (226.25, 1.5) {};
		\node [style=none] (464) at (227.75, 1.5) {};
		\node [style=X] (465) at (227.25, -1) {$1$};
		\node [style=Z] (466) at (225.75, -1) {};
		\node [style=scalar] (467) at (227.25, -0.25) {$a$};
		\node [style=scalarop] (468) at (225.75, -0.25) {$a$};
		\node [style=none] (469) at (226.75, -1.5) {};
		\node [style=none] (470) at (228.25, -1.5) {};
	\end{pgfonlayer}
	\begin{pgfonlayer}{edgelayer}
		\draw [in=90, out=-30] (459) to (461.center);
		\draw [in=-30, out=90] (462.center) to (460);
		\draw (459) to (463.center);
		\draw (460) to (464.center);
		\draw (469.center) to (461.center);
		\draw (470.center) to (462.center);
		\draw (465) to (467);
		\draw (466) to (468);
		\draw [in=-150, out=90] (468) to (459);
		\draw [in=-150, out=90] (467) to (460);
	\end{pgfonlayer}
\end{tikzpicture}
=
\begin{tikzpicture}
	\begin{pgfonlayer}{nodelayer}
		\node [style=X] (0) at (1.5, 0) {$a$};
		\node [style=none] (1) at (1, -1.5) {};
		\node [style=none] (2) at (1.5, -1.5) {};
		\node [style=none] (3) at (1, 1.5) {};
		\node [style=none] (4) at (1.5, 1.5) {};
	\end{pgfonlayer}
	\begin{pgfonlayer}{edgelayer}
		\draw (2.center) to (0);
		\draw (0) to (4.center);
		\draw (1.center) to (3.center);
	\end{pgfonlayer}
\end{tikzpicture}
\hspace*{.1cm}\in\alr_k
$$
\end{proof}
Therefore, we are justified in using string diagrams in $\alr_k$ to reason about morphisms in $\Aff\Lag\Rel_k$.

\cite{affine} interprets some components for electrical circuits  in terms of affine relations.  We translate these affine relations into composites of generators for Lagrangian relations.  This interpretation is also explored in \cite{passive,network}, albeit, not enjoying the graphical calculus for affine relations.
\begin{example}
\label{ex:circuits}
For any non-negative real $a$, wires, $a$-weighted resistors, inductors and capacitors have the following interpretations in $\Aff\Lag\Rel_{\mathbb{R}[x,y]/\langle xy-1\rangle}$:
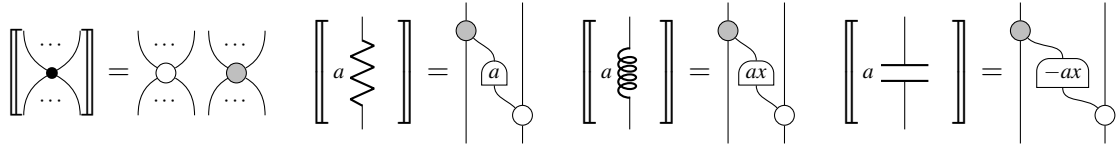
$$
\left\llbracket
\begin{tikzpicture}
	\begin{pgfonlayer}{nodelayer}
		\node [style=none] (0) at (21, 4.25) {};
		\node [style=none] (1) at (22, 4.25) {};
		\node [style=none] (2) at (21, 2.75) {};
		\node [style=none] (3) at (22, 2.75) {};
		\node [style=dot] (4) at (21.5, 3.5) {};
		\node [style=none] (5) at (21.5, 4) {$\cdots$};
		\node [style=none] (6) at (21.5, 3) {$\cdots$};
	\end{pgfonlayer}
	\begin{pgfonlayer}{edgelayer}
		\draw [in=150, out=-90] (0.center) to (4);
		\draw [in=90, out=-150] (4) to (2.center);
		\draw [in=-30, out=90] (3.center) to (4);
		\draw [in=-90, out=30] (4) to (1.center);
	\end{pgfonlayer}
\end{tikzpicture}
\right\rrbracket
=
\begin{tikzpicture}
	\begin{pgfonlayer}{nodelayer}
		\node [style=none] (501) at (237.5, 4.25) {};
		\node [style=none] (502) at (238.5, 4.25) {};
		\node [style=none] (503) at (237.5, 2.75) {};
		\node [style=none] (504) at (238.5, 2.75) {};
		\node [style=X] (505) at (238, 3.5) {};
		\node [style=none] (506) at (238, 4) {$\cdots$};
		\node [style=none] (507) at (238, 3) {$\cdots$};
		\node [style=none] (508) at (236.25, 4.25) {};
		\node [style=none] (509) at (237.25, 4.25) {};
		\node [style=none] (510) at (236.25, 2.75) {};
		\node [style=none] (511) at (237.25, 2.75) {};
		\node [style=Z] (512) at (236.75, 3.5) {};
		\node [style=none] (513) at (236.75, 4) {$\cdots$};
		\node [style=none] (514) at (236.75, 3) {$\cdots$};
	\end{pgfonlayer}
	\begin{pgfonlayer}{edgelayer}
		\draw [in=150, out=-90] (501.center) to (505);
		\draw [in=90, out=-150] (505) to (503.center);
		\draw [in=-30, out=90] (504.center) to (505);
		\draw [in=-90, out=30] (505) to (502.center);
		\draw [in=150, out=-90] (508.center) to (512);
		\draw [in=90, out=-150] (512) to (510.center);
		\draw [in=-30, out=90] (511.center) to (512);
		\draw [in=-90, out=30] (512) to (509.center);
	\end{pgfonlayer}
\end{tikzpicture}
\hspace*{.5cm}
\left\llbracket
\tikz \draw (0,0) to[R=$a$] (0,2);
\hspace*{,3cm}
\right\rrbracket
=
\begin{tikzpicture}
	\begin{pgfonlayer}{nodelayer}
		\node [style=Z] (0) at (23, 3.75) {};
		\node [style=X] (1) at (22, 5.25) {};
		\node [style=scalar] (2) at (22.5, 4.5) {$a$};
		\node [style=none] (3) at (22, 5.75) {};
		\node [style=none] (4) at (23, 5.75) {};
		\node [style=none] (5) at (23, 3.25) {};
		\node [style=none] (6) at (22, 3.25) {};
	\end{pgfonlayer}
	\begin{pgfonlayer}{edgelayer}
		\draw [in=-90, out=135] (0) to (2);
		\draw [in=315, out=90] (2) to (1);
		\draw (1) to (3.center);
		\draw (1) to (6.center);
		\draw (5.center) to (0);
		\draw (4.center) to (0);
	\end{pgfonlayer}
\end{tikzpicture}
\hspace*{.5cm}
\left\llbracket
\tikz \draw (0,0) to[L=$a$] (0,2);
\hspace*{,3cm}
\right\rrbracket
=
\begin{tikzpicture}
	\begin{pgfonlayer}{nodelayer}
		\node [style=Z] (0) at (23, 3.75) {};
		\node [style=X] (1) at (22, 5.25) {};
		\node [style=scalar] (2) at (22.5, 4.5) {$ax$};
		\node [style=none] (3) at (22, 5.75) {};
		\node [style=none] (4) at (23, 5.75) {};
		\node [style=none] (5) at (23, 3.25) {};
		\node [style=none] (6) at (22, 3.25) {};
	\end{pgfonlayer}
	\begin{pgfonlayer}{edgelayer}
		\draw [in=-90, out=135] (0) to (2);
		\draw [in=315, out=90] (2) to (1);
		\draw (1) to (3.center);
		\draw (1) to (6.center);
		\draw (5.center) to (0);
		\draw (4.center) to (0);
	\end{pgfonlayer}
\end{tikzpicture}
\hspace*{.5cm}
\left\llbracket
\tikz \draw (0,0) to[C=$a$] (0,2);
\hspace*{,3cm}
\right\rrbracket
=
\begin{tikzpicture}
	\begin{pgfonlayer}{nodelayer}
		\node [style=Z] (0) at (23.25, 3.75) {};
		\node [style=X] (1) at (21.75, 5.25) {};
		\node [style=scalar] (2) at (22.5, 4.5) {$-ax$};
		\node [style=none] (3) at (21.75, 5.75) {};
		\node [style=none] (4) at (23.25, 5.75) {};
		\node [style=none] (5) at (23.25, 3.25) {};
		\node [style=none] (6) at (21.75, 3.25) {};
	\end{pgfonlayer}
	\begin{pgfonlayer}{edgelayer}
		\draw [in=-90, out=135] (0) to (2);
		\draw [in=315, out=90] (2) to (1);
		\draw (1) to (3.center);
		\draw (1) to (6.center);
		\draw (5.center) to (0);
		\draw (4.center) to (0);
	\end{pgfonlayer}
\end{tikzpicture}
$$
Similarly for $a$-valued voltage and current sources (again, for  non-negative real number $a$):
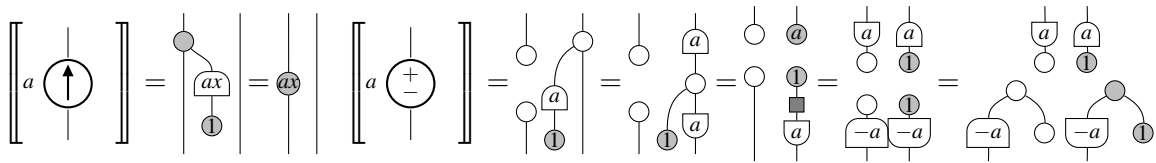
$$
\left\llbracket
\begin{tikzpicture}
	\begin{pgfonlayer}{nodelayer}
		\node [style=none] (0) at (0, 2) {};
		\node [style=isourceAMshape,rotate=90] (1) at (0, 1) {};
		\node [style=none] (2) at (0, 0) {};
	\end{pgfonlayer}
	\begin{pgfonlayer}{edgelayer}
		\draw (2.center) to (1);
		\draw (1) to (0.center);
		\node [style=none] (3) at (-.7, 1) {$a$};
	\end{pgfonlayer}
\end{tikzpicture}
\hspace*{,3cm}
\right\rrbracket
=
\begin{tikzpicture}
	\begin{pgfonlayer}{nodelayer}
		\node [style=X] (1) at (22, 5.25) {};
		\node [style=scalar] (2) at (22.5, 4.5) {$ax$};
		\node [style=none] (3) at (22, 5.75) {};
		\node [style=none] (4) at (23, 5.75) {};
		\node [style=none] (5) at (23, 3.25) {};
		\node [style=none] (6) at (22, 3.25) {};
		\node [style=X] (7) at (22.5, 3.75) {$1$};
	\end{pgfonlayer}
	\begin{pgfonlayer}{edgelayer}
		\draw [in=315, out=90] (2) to (1);
		\draw (1) to (3.center);
		\draw (1) to (6.center);
		\draw (7) to (2);
		\draw (5.center) to (4.center);
	\end{pgfonlayer}
\end{tikzpicture}
=
\begin{tikzpicture}
	\begin{pgfonlayer}{nodelayer}
		\node [style=X] (32) at (129.25, -0.25) {$ax$};
		\node [style=none] (33) at (129.25, 1) {};
		\node [style=none] (34) at (129.75, 1) {};
		\node [style=none] (35) at (129.75, -1.5) {};
		\node [style=none] (36) at (129.25, -1.5) {};
	\end{pgfonlayer}
	\begin{pgfonlayer}{edgelayer}
		\draw (32) to (33.center);
		\draw (32) to (36.center);
		\draw (35.center) to (34.center);
	\end{pgfonlayer}
\end{tikzpicture}
\hspace{,3cm}
\left\llbracket
\begin{tikzpicture}
	\begin{pgfonlayer}{nodelayer}
		\node [style=none] (0) at (0, 2) {};
		\node [style=vsourceAMshape,rotate=-90] (1) at (0, 1) {};
		\node [style=none] (2) at (0, 0) {};
		\node [style=none] (3) at (-.7, 1) {$a$};
	\end{pgfonlayer}
	\begin{pgfonlayer}{edgelayer}
		\draw (2.center) to (1);
		\draw (1) to (0.center);
	\end{pgfonlayer}
\end{tikzpicture}
\hspace*{,3cm}
\right\rrbracket
=
\begin{tikzpicture}
	\begin{pgfonlayer}{nodelayer}
		\node [style=none] (28) at (9, 1) {};
		\node [style=none] (29) at (10, 1) {};
		\node [style=none] (30) at (10, -1.5) {};
		\node [style=none] (31) at (9, -1.5) {};
		\node [style=Z] (32) at (9, 0.25) {};
		\node [style=Z] (33) at (9, -0.75) {};
		\node [style=X] (34) at (9.5, -1.25) {$1$};
		\node [style=scalar] (35) at (9.5, -0.5) {$a$};
		\node [style=Z] (36) at (10, 0.5) {};
	\end{pgfonlayer}
	\begin{pgfonlayer}{edgelayer}
		\draw (31.center) to (33);
		\draw (32) to (28.center);
		\draw (36) to (29.center);
		\draw [in=90, out=-150] (36) to (35);
		\draw (35) to (34);
		\draw (30.center) to (36);
	\end{pgfonlayer}
\end{tikzpicture}
=
\begin{tikzpicture}
	\begin{pgfonlayer}{nodelayer}
		\node [style=none] (37) at (11, 1) {};
		\node [style=none] (38) at (12, 1) {};
		\node [style=none] (39) at (12, -1.5) {};
		\node [style=none] (40) at (11, -1.5) {};
		\node [style=Z] (41) at (11, 0.25) {};
		\node [style=Z] (42) at (11, -0.75) {};
		\node [style=X] (43) at (11.5, -1.25) {$1$};
		\node [style=Z] (45) at (12, -0.25) {};
		\node [style=scalar] (46) at (12, 0.5) {$a$};
		\node [style=scalarop] (47) at (12, -1) {$a$};
	\end{pgfonlayer}
	\begin{pgfonlayer}{edgelayer}
		\draw (40.center) to (42);
		\draw (41) to (37.center);
		\draw (45) to (46);
		\draw (46) to (38.center);
		\draw [in=-150, out=90] (43) to (45);
		\draw (39.center) to (47);
		\draw (47) to (45);
	\end{pgfonlayer}
\end{tikzpicture}
=
\begin{tikzpicture}
	\begin{pgfonlayer}{nodelayer}
		\node [style=none] (73) at (17, 0.75) {};
		\node [style=none] (74) at (17.75, 0.75) {};
		\node [style=none] (75) at (17.75, -2) {};
		\node [style=none] (76) at (17, -2) {};
		\node [style=Z] (77) at (17, 0.25) {};
		\node [style=Z] (78) at (17, -0.5) {};
		\node [style=X] (79) at (17.75, 0.25) {$a$};
		\node [style=scalarop] (80) at (17.75, -1.5) {$a$};
		\node [style=X] (83) at (17.75, -0.5) {$1$};
		\node [style=s] (84) at (17.75, -1) {};
	\end{pgfonlayer}
	\begin{pgfonlayer}{edgelayer}
		\draw (76.center) to (78);
		\draw (77) to (73.center);
		\draw [in=-90, out=90] (75.center) to (80);
		\draw (79) to (74.center);
		\draw (84) to (83);
		\draw [in=-90, out=90] (80) to (84);
	\end{pgfonlayer}
\end{tikzpicture}
=
\begin{tikzpicture}
	\begin{pgfonlayer}{nodelayer}
		\node [style=none] (84) at (18.75, 1.25) {};
		\node [style=none] (85) at (19.5, 1.25) {};
		\node [style=none] (86) at (19.5, -1.5) {};
		\node [style=none] (87) at (18.75, -1.5) {};
		\node [style=Z] (88) at (18.75, 0.25) {};
		\node [style=Z] (89) at (18.75, -0.5) {};
		\node [style=X] (90) at (19.5, 0.25) {$1$};
		\node [style=X] (91) at (19.5, -0.5) {$1$};
		\node [style=scalarop] (92) at (19.5, -1) {$-a$};
		\node [style=scalar] (93) at (18.75, -1) {$-a$};
		\node [style=scalar] (94) at (19.5, 0.75) {$a$};
		\node [style=scalarop] (95) at (18.75, 0.75) {$a$};
	\end{pgfonlayer}
	\begin{pgfonlayer}{edgelayer}
		\draw (86.center) to (92);
		\draw (92) to (91);
		\draw (89) to (93);
		\draw (87.center) to (93);
		\draw (88) to (95);
		\draw (95) to (84.center);
		\draw (90) to (94);
		\draw (94) to (85.center);
	\end{pgfonlayer}
\end{tikzpicture}
=
\begin{tikzpicture}
	\begin{pgfonlayer}{nodelayer}
		\node [style=none] (602) at (268.25, 1.75) {};
		\node [style=none] (603) at (269, 1.75) {};
		\node [style=none] (604) at (269, -1) {};
		\node [style=none] (605) at (267.25, -1) {};
		\node [style=Z] (606) at (268.25, 0.75) {};
		\node [style=Z] (607) at (267.75, 0.25) {};
		\node [style=X] (608) at (269, 0.75) {$1$};
		\node [style=X] (609) at (269.5, 0.25) {};
		\node [style=scalarop] (610) at (269, -0.5) {$-a$};
		\node [style=scalar] (611) at (267.25, -0.5) {$-a$};
		\node [style=scalar] (612) at (269, 1.25) {$a$};
		\node [style=scalarop] (613) at (268.25, 1.25) {$a$};
		\node [style=Z] (614) at (268.25, -0.5) {};
		\node [style=X] (615) at (270, -0.5) {$1$};
	\end{pgfonlayer}
	\begin{pgfonlayer}{edgelayer}
		\draw (604.center) to (610);
		\draw [in=-150, out=90] (610) to (609);
		\draw [in=90, out=-150] (607) to (611);
		\draw (605.center) to (611);
		\draw (606) to (613);
		\draw (613) to (602.center);
		\draw (608) to (612);
		\draw (612) to (603.center);
		\draw [in=-30, out=90] (615) to (609);
		\draw [in=-30, out=90] (614) to (607);
	\end{pgfonlayer}
\end{tikzpicture}
$$
\end{example}
Note that these generators do not generate the whole category of Lagrangian relations; for instance, the coefficients are required to be non-negative.

\subsection{Stabilizer circuits and Spekkens' toy model}

In this subsection, we show that, when $p$ is an odd prime, the prop of affine Lagrangian relations over $\F_p$  is isomoprhic to $p$-dimensional qudit stabilizer circuits, modulo invertible scalars.  We first consider an intermediary fragment between the Fourier-free, phase-free fragments and stabilizer circuits.

\begin{definition}
The qudit {\bf boost operator} is the following unitary on $d$ in $\Mat(\C)$, 
$
{\cal X} := \sum_{a =0}^{d-1} |  a+1  \rangle\langle a |
$.
\end{definition}

In the qubit case, the boost operator is just the NOT gate.  Adding the affine shift to $\ih_{\F_p}$ corresponds to adding the boost gate to the  Fourier-free, phase-free ZX-calculus, extending Lemma \ref{lemma:phasefree}.  This is a qudit generalization of the observation made in \cite{cole}:

\begin{lemma}
For $p$ prime, $\aih_{\F_p}$ is isomorphic as a prop  to the Fourier-free, $p$-dimensional qudit ZX-calculus with the boost operator modulo invertible scalars.
\end{lemma}

We can go further with affine Lagrangian relations.  Inspired by the work of Spekkens \cite{spekkens,spekkens2016quasi}:

\begin{definition}
When $p$ is prime, let {\bf Spekkens' qudit toy model} of dimension $p$  be the prop $\Aff\Lag\Rel_{\F_p}$.
\end{definition}

We first give a short review of the qudit stabilizer formalism, before establishing the equivalence between Spekkens' toy model and stabilizer circuits in the odd prime qudit case.  All of the material from Definition \ref{definition:begin} to \ref{lemma:end} are contained in \cite{generators}.

\begin{definition}
\label{definition:begin}
The qudit {\bf shift operator} is the following unitary on $d$ in $\Mat(\C)$, 
${\cal Z} := \sum_{a =0}^{d-1} e^{2\pi i a/d} |  a  \rangle\langle a |$.

The  $n$-qudit {\bf Pauli group} ${\frak P}_d^{\otimes n}$ is defined to be the subgroup of $\Mat(\C)(d^n,d^n)$ generated by the shift and boost operators as well as $I_d$ and the scalar $e^{\pi i /d}$ under tensor product and matrix multiplication.

An $n$-qudit {\bf Clifford operator} $U$ is an $d^n$-dimensional unitary such that $U {\frak P}_d^{ n} U^\dag = {\frak P}_d^{ n}$.

The $n$-qudit {\bf Clifford group} is formed by the $n$-qudit Clifford operators under matrix multiplication.

An $n$-qudit {\bf stabilizer state} is a state $ U |0\rangle^{\otimes n}$ for an $n$-qudit Clifford $U$.

Given any $n$-qudit stabilizer state $|\psi \rangle$,  the {\bf stabilizer group} of $|\psi \rangle$  is the (Abelian) subgroup of ${\frak S}_{|\psi\rangle} \subset {\frak P}_d^{ n}$ whose elements are the $U \in {\frak P}_d^{ n}$ for which $U|\psi \rangle=|\psi \rangle$.
\end{definition}

\begin{lemma}\label{lem:stabphase}
Two stabilizer states with the same stabilizer groups are the same, up to global phases.
\end{lemma}

\begin{lemma}
\label{lemma:end}
For natural numbers $n,d \geq 2$ the $n$-dimensional qudit stabilizer group modulo invertible scalars is generated under tensor and composition of $I_d$ as well as the boost operator ${\cal X}$, the controlled-boost operator  ${\cal C}$, the Fourier transform ${\cal F}$ and the phase-shift operator ${\cal S}$:
$$
{\cal C}  := \sum_{a,b = 0 }^{d-1} |a,a+b \rangle\langle a, b|
\hspace*{.5cm}
{\cal F}  := \frac{1}{\sqrt d}\sum_{a,b= 0 }^{d-1} e^{2\pi i ab/d} |b \rangle\langle a|
\hspace*{.5cm}
{\cal S} := \sum_{a =0}^{d-1} e^{\pi i a (a+d)/d} |  a  \rangle\langle a |
$$
Notice that the boost operator can be obtained by ${\cal Z }={\cal F}{\cal X} {\cal F}^{2}$.
\end{lemma}

\begin{definition}
Let $\Stab_p$ denote the subcategory of $\Mat(\C)$ generated by the $p$-dimensional qudit Clifford group as well as the vectors $| 0\rangle$, $\langle 0|$, quotiented by invertible scalars.
\end{definition}

The following isomorphism is described in \cite{gross}, when restricted to the nonempty case.  This comes from the projective representation of the $n$ qudit odd-prime-dimensional Clifford group in terms of the affine symplectomorphisms over $\F_p^n$.  However, since there is only one empty relation and one zero matrix of every type, we get the following result immediately:

\begin{lemma}
For every odd prime $p$, there is an isomorphism $G:\Aff\Lag\Rel_{\F_p}(\F_p^0, \F_p^{2n}) \to \Stab_p(0,n)$ determined by:
$$
\begin{tikzpicture}
	\begin{pgfonlayer}{nodelayer}
		\node [style=X] (393) at (206, -2) {$\pi$};
	\end{pgfonlayer}
\end{tikzpicture}
\mapsto 
0
\hspace*{.5cm}
\begin{tikzpicture}
	\begin{pgfonlayer}{nodelayer}
		\node [style=X] (389) at (205, -2) {};
		\node [style=none] (390) at (205, -1.5) {};
		\node [style=Z] (391) at (204.5, -2) {};
		\node [style=none] (392) at (204.5, -1.5) {};
	\end{pgfonlayer}
	\begin{pgfonlayer}{edgelayer}
		\draw (389) to (390.center);
		\draw (391) to (392.center);
	\end{pgfonlayer}
\end{tikzpicture}
 \mapsto |0\rangle
\hspace*{.5cm}
\begin{tikzpicture}
	\begin{pgfonlayer}{nodelayer}
		\node [style=none] (391) at (206, -1) {};
		\node [style=none] (392) at (206, -2) {};
		\node [style=none] (393) at (206.5, -1) {};
		\node [style=none] (394) at (206.5, -2) {};
		\node [style=X] (395) at (206.5, -1.5) {$1$};
	\end{pgfonlayer}
	\begin{pgfonlayer}{edgelayer}
		\draw (392.center) to (391.center);
		\draw (393.center) to (395);
		\draw (395) to (394.center);
	\end{pgfonlayer}
\end{tikzpicture}
 \mapsto {\cal X}
\hspace*{.5cm}
C_1 \mapsto {\cal C}
\hspace*{.5cm}
F \mapsto {\cal F}
\hspace*{.5cm}
S_1 \mapsto {\cal S}
$$
\end{lemma}

We extend this isomorphim of states to an isomorphism of props:

\begin{theorem}
\label{theorem:spekkens}
When $p$ is an odd prime, the mapping $H:\Aff\Lag\Rel_{\F_p} \to \Stab_p$ defined by:
$$
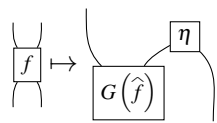
\begin{tikzpicture}
	\begin{pgfonlayer}{nodelayer}
		\node [style=map] (21) at (2, -2) {$f$};
		\node [style=none] (22) at (1.75, -1.25) {};
		\node [style=none] (23) at (2.25, -1.25) {};
		\node [style=none] (24) at (1.75, -2.75) {};
		\node [style=none] (25) at (2.25, -2.75) {};
	\end{pgfonlayer}
	\begin{pgfonlayer}{edgelayer}
		\draw [in=-90, out=120] (21) to (22.center);
		\draw [in=90, out=-120] (21) to (24.center);
		\draw [in=-60, out=90] (25.center) to (21);
		\draw [in=-90, out=60] (21) to (23.center);
	\end{pgfonlayer}
\end{tikzpicture}
\mapsto
\begin{tikzpicture}
	\begin{pgfonlayer}{nodelayer}
		\node [style=map] (7) at (12, 0) {$G\left(\hat f\right)$};
		\node [style=none] (8) at (11.25, 1.5) {};
		\node [style=map] (10) at (13, 1) {$\eta$};
		\node [style=none] (12) at (13.5, -0.5) {};
	\end{pgfonlayer}
	\begin{pgfonlayer}{edgelayer}
		\draw [in=135, out=-90] (8.center) to (7);
		\draw [in=-150, out=60, looseness=0.75] (7) to (10);
		\draw [in=-45, out=90] (12.center) to (10);
	\end{pgfonlayer}
\end{tikzpicture}
$$
is a symmetric monoidal equivalence, where $\eta$ is the cap of the compact closed structure induced by the $Z$ observable.

\end{theorem}

The main difficulty in proving this is to show that $H$ is a functor.  Lemma \ref{lem:stabphase} is important for showing that composition is well defined.

As we mentioned in the introduction, this is a categorical reformulation of the result of Spekkens' in which he shows that odd-prime-dimensional `quadrature epistricted theories' are operationally equivalent to prime-dimensional qudit stabilizer circuits \cite{spekkens2016quasi}.

A complete presentation for Spekkens' qubit toy model in terms of a category of relations was given \cite{backensspek} in a style which mirrors that of the qubit ZX-calculus~\cite{coecke2008interacting}. We now show how the generators of that presentation appear in our `doubled' formulation.

\begin{remark}
We can present Spekkens' $p$-dimensional qudit toy model in a manner similar to the ZX-calculus, in terms of being generated by spiders with phases labelled by the group $\Z/p\Z\times \Z/p\Z$,:
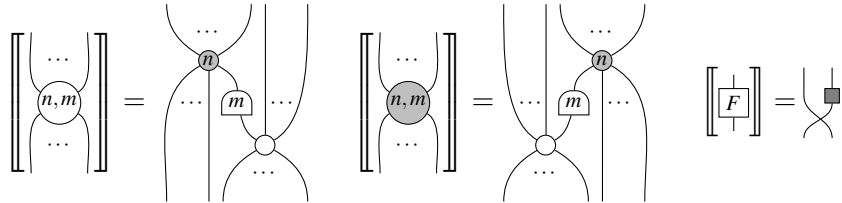
$$
\left\llbracket
\begin{tikzpicture}
	\begin{pgfonlayer}{nodelayer}
		\node [style=none] (0) at (21, 5) {};
		\node [style=none] (1) at (22, 5) {};
		\node [style=none] (2) at (21, 2.5) {};
		\node [style=none] (3) at (22, 2.5) {};
		\node [style=Z] (4) at (21.5, 3.75) {$\hspace*{.05cm}n,m\hspace*{.05cm}$};
		\node [style=none] (5) at (21.5, 4.5) {$\cdots$};
		\node [style=none] (6) at (21.5, 3) {$\cdots$};
		\node [style=none] (7) at (21.5, 4.75) {};
		\node [style=none] (8) at (21.5, 2.75) {};
	\end{pgfonlayer}
	\begin{pgfonlayer}{edgelayer}
		\draw [in=150, out=-90, looseness=0.75] (0.center) to (4);
		\draw [in=90, out=-150, looseness=0.75] (4) to (2.center);
		\draw [in=-30, out=90, looseness=0.75] (3.center) to (4);
		\draw [in=-90, out=30, looseness=0.75] (4) to (1.center);
	\end{pgfonlayer}
\end{tikzpicture}
\right\rrbracket
=
\begin{tikzpicture}
	\begin{pgfonlayer}{nodelayer}
		\node [style=none] (9) at (231.75, 0.5) {};
		\node [style=none] (10) at (231.75, -3) {};
		\node [style=Z] (11) at (231, -2) {};
		\node [style=none] (12) at (231.32, -1.25) {$\cdots$};
		\node [style=none] (13) at (231, -2.5) {$\cdots$};
		\node [style=none] (14) at (230.75, 0.5) {};
		\node [style=none] (15) at (229.25, -3) {};
		\node [style=X] (16) at (230, -0.5) {$n$};
		\node [style=none] (17) at (230, 0) {$\cdots$};
		\node [style=none] (18) at (229.72, -1.25) {$\cdots$};
		\node [style=none] (19) at (230, -3) {};
		\node [style=none] (20) at (230.25, -3) {};
		\node [style=none] (21) at (229.25, 0.5) {};
		\node [style=none] (22) at (231, 0.5) {};
		\node [style=scalar] (23) at (230.5, -1.25) {$m$};
		\node [style=none] (24) at (230, 0.25) {};
		\node [style=none] (25) at (231, -2.75) {};
		\node [style=none] (26) at (229.7, -1.5) {};
		\node [style=none] (27) at (231.3, -1) {};
	\end{pgfonlayer}
	\begin{pgfonlayer}{edgelayer}
		\draw [in=-30, out=90] (10.center) to (11);
		\draw [in=-90, out=30, looseness=0.75] (11) to (9.center);
		\draw [in=90, out=-150] (16) to (15.center);
		\draw [in=-90, out=30] (16) to (14.center);
		\draw (19.center) to (16);
		\draw [in=-150, out=90] (20.center) to (11);
		\draw [in=-90, out=150] (16) to (21.center);
		\draw (11) to (22.center);
		\draw [in=-75, out=150] (11) to (23);
		\draw [in=330, out=90] (23) to (16);
	\end{pgfonlayer}
\end{tikzpicture}
\hspace*{.5cm}
\left\llbracket
\begin{tikzpicture}
	\begin{pgfonlayer}{nodelayer}
		\node [style=none] (0) at (21, 5) {};
		\node [style=none] (1) at (22, 5) {};
		\node [style=none] (2) at (21, 2.5) {};
		\node [style=none] (3) at (22, 2.5) {};
		\node [style=X] (4) at (21.5, 3.75) {$\hspace*{.05cm}n,m\hspace*{.05cm}$};
		\node [style=none] (5) at (21.5, 4.5) {$\cdots$};
		\node [style=none] (6) at (21.5, 3) {$\cdots$};
		\node [style=none] (7) at (21.5, 4.75) {};
		\node [style=none] (8) at (21.5, 2.75) {};
	\end{pgfonlayer}
	\begin{pgfonlayer}{edgelayer}
		\draw [in=150, out=-90, looseness=0.75] (0.center) to (4);
		\draw [in=90, out=-150, looseness=0.75] (4) to (2.center);
		\draw [in=-30, out=90, looseness=0.75] (3.center) to (4);
		\draw [in=-90, out=30, looseness=0.75] (4) to (1.center);
	\end{pgfonlayer}
\end{tikzpicture}
\right\rrbracket
=
\begin{tikzpicture}
	\begin{pgfonlayer}{nodelayer}
		\node [style=none] (0) at (232.75, 0.5) {};
		\node [style=none] (1) at (232.75, -3) {};
		\node [style=Z] (2) at (233.5, -2) {};
		\node [style=none] (3) at (233.25, -1.25) {$\cdots$};
		\node [style=none] (4) at (233.5, -2.5) {$\cdots$};
		\node [style=none] (5) at (233.75, 0.5) {};
		\node [style=none] (6) at (235.25, -3) {};
		\node [style=X] (7) at (234.5, -0.5) {$n$};
		\node [style=none] (8) at (234.5, 0) {$\cdots$};
		\node [style=none] (9) at (234.82, -1.25) {$\cdots$};
		\node [style=none] (10) at (234.5, -3) {};
		\node [style=none] (11) at (234.25, -3) {};
		\node [style=none] (12) at (235.25, 0.5) {};
		\node [style=none] (13) at (233.5, 0.5) {};
		\node [style=scalar] (14) at (234, -1.25) {$m$};
		\node [style=none] (15) at (233.25, -1) {};
		\node [style=none] (16) at (234.5, 0.25) {};
		\node [style=none] (17) at (233.5, -2.75) {};
		\node [style=none] (18) at (234.78, -1.5) {};
	\end{pgfonlayer}
	\begin{pgfonlayer}{edgelayer}
		\draw [in=-150, out=90] (1.center) to (2);
		\draw [in=-90, out=150, looseness=0.75] (2) to (0.center);
		\draw [in=90, out=-30] (7) to (6.center);
		\draw [in=-90, out=150] (7) to (5.center);
		\draw (10.center) to (7);
		\draw [in=-30, out=90] (11.center) to (2);
		\draw [in=-90, out=30] (7) to (12.center);
		\draw (2) to (13.center);
		\draw [in=-105, out=30] (2) to (14);
		\draw [in=-150, out=90] (14) to (7);
	\end{pgfonlayer}
\end{tikzpicture}
\hspace*{.5cm}
\left\llbracket
\begin{tikzpicture}
	\begin{pgfonlayer}{nodelayer}
		\node [style=none] (0) at (1.25, -1) {};
		\node [style=map] (1) at (1.25, -1.5) {$F$};
		\node [style=none] (2) at (1.25, -2) {};
	\end{pgfonlayer}
	\begin{pgfonlayer}{edgelayer}
		\draw (2.center) to (1);
		\draw (1) to (0.center);
	\end{pgfonlayer}
\end{tikzpicture}
\right\rrbracket
=
\begin{tikzpicture}
	\begin{pgfonlayer}{nodelayer}
		\node [style=none] (0) at (0.5, 1) {};
		\node [style=none] (1) at (0.5, -0.25) {};
		\node [style=none] (2) at (1, -0.25) {};
		\node [style=none] (3) at (1, 1) {};
		\node [style=s] (4) at (1, 0.5) {};
		\node [style=none] (5) at (0.5, 0.5) {};
	\end{pgfonlayer}
	\begin{pgfonlayer}{edgelayer}
		\draw (4) to (3.center);
		\draw [in=90, out=-90] (4) to (1.center);
		\draw [in=-90, out=90] (2.center) to (5.center);
		\draw (5.center) to (0.center);
	\end{pgfonlayer}
\end{tikzpicture}
$$
The Fourier transform is redundant, as it can be obtained by Euler decomposition.
\end{remark}
Notice that the phases of the $Z$ and $X$ observables are elements $(n,m)$ of $\F_p \times \F_p$, and it is easy to see how the doubled spiders satisfy the phased spider fusion laws with respect to the group $\Z/p\Z\times \Z/p\Z$, as discussed in \cite[p.\ 166]{ranchin2016alternative}.  As discussed in \cite{coecke2011phase} this is one of the central features which separates Spekkens' qubit model from qubit stabilizers, whose phase group is $\Z/4\Z$.  This fact can be also observed graphically in terms of the stabilizer fragment of the ZX-calculus (in contrast to the presentation of Spekkens' qubit toy model) which also enjoys a complete axiomatization \cite{backensspek}.

By stating the interpretations of Spekkens' toy model in terms of the graphical calculus for Lagrangian relations alongside that of electrical circuits, we see the evident analogy between the phases in the ZX-calculus and the resistors, inductor, capacitors and voltage sources in electrical circuits.

\section{Further work}

There are several directions in which the work in this paper could be further explored. Since linear relations  can be defined over a principal ideal domain over a field, it is natural to ask if the work can be generalized to this setting. Also, we have not given a completeness result entirely in terms of the generators of $\Lag\Rel_k$. The proof of such would almost certainly involve mimicking the universality proofs of the qubit stabilizer/qutrit stabilizer/Spekkens' toy model \cite{backensstab,backensspek,qutrit} involving local equivalency/local complementation of graph states.  
If this were generalized to affine Lagrangian relations this would yield a proper completeness result for the odd-prime-dimensional qudit stabilizer ZX-calculus as a corollary. One could also potentially adapt this approach to characterize Lagrangian spans as described in \cite[p.\ 187]{fong2016algebra}, where the scalars are not all quotiented out.  Perhaps this would give a semantics for odd-prime-dimensional qudit stabilizer circuits on the nose.

This paper illuminates the deep connection between stabilizer circuits and electrical circuits.  Perhaps, this can be taken further by adding nonlinear generators.  For example, by doubling again, one could also consider discarding as a generator, as in \cite{disc}.  Does the discard in stabilizer quantum mechanics obey analagous equations to the ground in electrical circuits?

\nocite{coecke2008interacting}
\nocite{ihpub}
\nocite{niel}

\bibliographystyle{eptcs}

\bibliography{lagrel}

\end{document}